\documentclass{article}
\usepackage[utf8]{inputenc}
\usepackage{amsmath, amsthm, amssymb}
\usepackage{enumerate}
\usepackage{authblk}
\usepackage{thmtools}
\usepackage{mathrsfs}
\usepackage{thm-restate}
\usepackage{soul}
\usepackage{xcolor}
\usepackage{tikz}
\usetikzlibrary{matrix,decorations.pathreplacing, calc, positioning,fit,shapes.geometric}
\usepackage{fullpage}
\usepackage{thm-restate}
\usepackage{hyperref}

\theoremstyle{definition}
\newtheorem{theorem}{Theorem}[section]

\newtheorem{corollary}[theorem]{Corollary}

\newtheorem{lemma}[theorem]{Lemma}
\newtheorem{remark}[theorem]{Remark}

\newtheorem{observation}[theorem]{Observation}

\newtheorem{conjecture}[theorem]{Conjecture}
\newtheorem{claim}[theorem]{Claim}
\newtheorem{question}[theorem]{Question}
\theoremstyle{remark}

\usepackage{comment}

\def\A{\mathcal{A}} 
\newcommand{\NN}{\mathbb{N}} 
\def\S{\mathcal{S}} 

\usetikzlibrary{intersections}

\usetikzlibrary{patterns}
\usetikzlibrary{patterns.meta}
\usetikzlibrary{calc}
\def\centerarc[#1](#2)(#3:#4:#5)
    { \draw[#1] ($(#2)+({#5*cos(#3)},{#5*sin(#3)})$) arc (#3:#4:#5) }
\def\nodearc(#1)(#2)(#3:#4)[#5](#6)
    { \node (#1) at ($(#2)+({#4*cos(#3)},{#4*sin(#3)})$) [#5] {#6}}
\def\nodeellipse(#1)(#2)(#3:#4:#5)[#6](#7)
    { \node (#1) at ($(#2)+({#4*cos(#3)},{#5*sin(#3)})$) [#6] {#7}}

\pgfdeclarepattern{
  name=hatch,
  parameters={\hatchsize,\hatchangle,\hatchlinewidth},
  bottom left={\pgfpoint{-.1pt}{-.1pt}},
  top right={\pgfpoint{\hatchsize+.1pt}{\hatchsize+.1pt}},
  tile size={\pgfpoint{\hatchsize}{\hatchsize}},
  tile transformation={\pgftransformrotate{\hatchangle}},
  code={
    \pgfsetlinewidth{\hatchlinewidth}
    \pgfpathmoveto{\pgfpoint{-.1pt}{-.1pt}}
    \pgfpathlineto{\pgfpoint{\hatchsize+.1pt}{\hatchsize+.1pt}}
    \pgfusepath{stroke}
  }
}

\tikzset{
  hatch size/.store in=\hatchsize,
  hatch angle/.store in=\hatchangle,
  hatch line width/.store in=\hatchlinewidth,
  hatch size=5pt,
  hatch angle=0pt,
  hatch line width=.5pt,
}

\usepackage{graphicx} 

\title{Reconfiguration of plane trees in convex geometric graphs\footnote{The authors are supported by ANR project GrR (ANR-18-CE40-0032).}}
\author{Nicolas Bousquet, Lucas De Meyer, Théo Pierron, Alexandra Wesolek}
\affil{Université de Lyon, LIRIS, CNRS, Université Claude Bernard Lyon 1}
\date{\today}

\begin{document}

\maketitle

\begin{abstract}
A non-crossing spanning tree of a set of points in the plane is a spanning tree whose edges pairwise do not cross. Avis and Fukuda in 1996 proved that there always exists a flip sequence of length at most $2n-4$ between any pair of non-crossing spanning trees (where $n$ denotes the number of points). Hernando et al. proved that the length of a minimal flip sequence can be of length at least $\frac 32 n$.
Two recent results of Aichholzer et al. and Bousquet et al. improved the Avis and Fukuda upper bound by proving that there always exists a flip sequence of length respectively at most $2n - \log n$ and $2n - \sqrt{n}$.

We improve the upper bound by a linear factor for the first time in 25 years by proving that there always exists a flip sequence between any pair of non-crossing spanning trees $T_1,T_2$ of length at most $c n$ where $c \approx 1.95$.  Our result is actually stronger since  we prove that, for any two trees $T_1,T_2$, there exists a flip sequence from $T_1$ to $T_2$ of length at most $c |T_1 \setminus T_2|$.

We also improve the best lower bound in terms of the symmetric difference by proving that there exists a pair of trees $T_1,T_2$ such that a minimal flip sequence has length $\frac 53 |T_1 \setminus T_2|$, improving the lower bound of Hernando et al. by considering the symmetric difference instead of the number of vertices.

We generalize this lower bound construction to non-crossing flips (where we close the gap between upper and lower bounds) and rotations.
\end{abstract}

\section{Introduction}
Let $C$ be a set of $n$ points in the plane in convex position. A \emph{spanning tree} $T$ on the set of points $C$ is a subset of edges that forms a connected acyclic graph on $C$. A spanning tree $T$ on $C$ is \emph{non-crossing} if every pair of edges of $T$ (represented by the straight line interval between their endpoints) are pairwise non-crossing.

Let us denote by $\mathcal{S}(C)$ the set of all non-crossing spanning trees on the point set $C$. Let $T\in \mathcal{S}(C)$.
A \emph{flip} on $T$ consists of removing an edge $e$ from $T$ and adding another edge $f$ so that the resulting graph $(T \cup f) \setminus e$ is also a spanning tree. A \emph{flip sequence} is a sequence of non-crossing spanning trees such that consecutive spanning trees in the sequence differ by exactly one flip. 
Equivalently, one can define the \emph{configuration graph} on the vertex set $\mathcal{S}(C)$ where two trees $T,T'$ are adjacent if they differ in exactly one edge (that is $|T \setminus T'|=|T'\setminus T|=1$). A (minimal) flip sequence is a (shortest) path in the configuration graph.

\subsection{Flips between non-crossing spanning trees.}

Avis and Fukuda~\cite{avis1996} proved that there always exists a flip sequence between any pair of non-crossing spanning trees of length at most $2n-4$ by showing that there is a star\footnote{A \emph{star} is a spanning tree with at most one vertex of degree at least $2$.} $S$ on $C$ such that $T_1$ and $T_2$ can be turned into $S$ with at most $n-2$ flips. In fact, they showed that this flip sequence exists even if the point set $C$ is in general  position.

Given two spanning trees $T_1,T_2$, the \emph{symmetric difference} between $T_1$ and $T_2$ is denoted by $\Delta(T_1,T_2)= (T_1 \setminus T_2) \cup (T_2 \setminus T_1)$. We denote by $\delta(T_1,T_2)=|\Delta(T_1,T_2)|/2$ the number of edges in $T_1$ and not in $T_2$, which is a trivial lower bound on the length of a flip sequence from $T_1$ to $T_2$.

It is well-known that the set of spanning trees of a graph $G$ forms a matroid. In particular, for any possible pair of spanning trees $T_1,T_2$, there is a (non geometric) flip sequence that transforms $T_1$ into $T_2$ in exactly $\delta(T_1,T_2)$ flips. So if we do not care about geometric properties of the representation of the spanning trees, it is always possible to transform a spanning tree $T_1$ into $T_2$ using at most $n-1$ flips. One can wonder if the same holds if we want to keep non-crossing spanning trees all along the flip sequence. Hernando et al.~\cite{hernando1999geotree} answered this question in the negative by providing, for every $n$, two non-crossing spanning trees $T_1,T_2$ on a convex set of $n$ points whose minimal flip sequence needs $\frac{3}{2}n -5$ flips (we give their example in Figure~\ref{fig:hernando}).

\begin{figure}[hbtp]
        \begin{center}
        \tikzstyle{vertex}=[circle,draw, minimum size=7pt, scale=0.5, inner sep=1pt, fill = black]
        \tikzstyle{fleche}=[->,>=latex]
        \tikzstyle{labell}=[text opacity=1, scale =0.9]
        \begin{tikzpicture}[scale=1.2]

         \begin{scope}
            \clip (0,0) circle (1.5cm);
        \end{scope}
        
        \nodearc(a1)(0,0)(0:1.5)[vertex]();
        \nodearc(a2)(0,0)(36:1.5)[vertex]();
        \nodearc(a4)(0,0)(72:1.5)[vertex]();
        \nodearc(a5)(0,0)(108:1.5)[vertex]();
        \nodearc(a6)(0,0)(144:1.5)[vertex]();
        \nodearc(b1)(0,0)(180:1.5)[vertex]();
        \nodearc(b2)(0,0)(-36:1.5)[vertex]();
        \nodearc(b4)(0,0)(-72:1.5)[vertex]();
        \nodearc(b5)(0,0)(-108:1.5)[vertex]();
        \nodearc(b6)(0,0)(-144:1.5)[vertex]();

        \centerarc[thick](0,0)(-144:-180:1.5);
        \centerarc[thick](0,0)(0:36:1.5);

        \draw[thick] (a1) to (a4);
        \draw[thick] (a1) to (a5);
        \draw[thick] (a1) to (a6);
        \draw[thick] (a1) to (b1);
        \draw[thick] (b1) to (b2);
        \draw[thick] (b1) to (b4);
        \draw[thick] (b1) to (b5);

        \draw[thick,red] (b6) to[bend right=5] (b1);
        \draw[thick,red] (b6) to (a4);
        \draw[thick,red] (b6) to (a5);
        \draw[thick,red] (b6) to (a6);
        \draw[thick,red] (a2) to[bend right=5] (a1);
        \draw[thick,red] (a2) to (b2);
        \draw[thick,red] (a2) to (b4);
        \draw[thick,red] (a2) to (b5);
        \draw[thick,red] (a2) to (b6);

        \end{tikzpicture}
        \end{center}
        \caption{A minimal flip sequence between $T_1$ (in black) and $T_2$ (in red) has length exactly $\lfloor \frac{3}{2}n \rfloor - 5 = 10$.}
        \label{fig:hernando}
\end{figure}
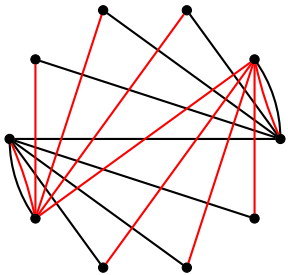

During 25 years, no improvement of the lower or upper bound has been obtained until a recent result of Aichholzer et al.~\cite{aichholzer2022reconfiguration}. They showed that the upper bound of Avis and Fukuda can be improved when points are in convex position by proving that there exists a flip sequence between any pair of non-crossing spanning trees of length at most $2n-\Omega(\log n)$. Their result has been further improved by Bousquet et al.~\cite{bousquet2023note} who proved that $2n-\Omega(\sqrt{n})$ flips are enough. However, until now, there does not exist any general proof that there always exists a flip sequence of length at most $(2-\epsilon)n$ for some $\epsilon >0$.

In both papers, the authors prove as well the existence of shorter flip sequences when one (or both) of the trees has a special shape. Aichholzer et al~\cite{aichholzer2022reconfiguration} proved that when the points are in convex position and $T_1$ is a path then there exists a flip sequence of length at most $\frac32n - 2 - |T_1 \cap T_2| = \frac{n + |\Delta(T_1,T_2)|}{2} - 1$. Bousquet et al.~\cite{bousquet2023note} proved that there exists a flip sequence of length at most $\frac 32 n$ when the points are in convex position and $T_1$ is a path or a nice caterpillar\footnote{A \emph{caterpillar} is a tree such that the set of nodes that are not leaves induces a path. Without giving the exact definition, a nice caterpillar is a caterpillar such that every chord cuts in a nice way the geometric representation.}.

Bousquet et al.~\cite{bousquet2023note} conjectured that the lower bound of Hernando et al.~\cite{hernando1999geotree} is essentially tight: 
\begin{conjecture}
\label{conj:upperbound}
    Let $C$ be a set of $n$ points in convex position. There exists a flip sequence between any pair of non-crossing spanning trees of length at most $\frac{3}{2} n$. 
\end{conjecture}

One can easily prove that there exists a flip sequence of length at most $2 \delta(T_1,T_2)$ between any pair of non-crossing spanning trees in convex position. The improvement of Aichholzer et al.~\cite{aichholzer2022reconfiguration} also holds in that setting. Since in the example of Hernando et al. the intersection is reduced to two edges, one can wonder if Conjecture~\ref{conj:upperbound} can be extended to the symmetric difference, namely:

\begin{conjecture}
\label{conj:upperbound_symdif}
    Let $C$ be a set of $n$ points in convex position. There exists a flip sequence between any pair of non-crossing spanning trees $T_1,T_2$ of length at most $\frac{3}{2} \delta(T_1,T_2)$. 
\end{conjecture}

Our main results, discussed in more details in the next paragraphs, first consist in (i) improving the best known upper bound by breaking the linear factor $2$ of the threshold on the length of a minimal flip sequence (even in terms of the symmetric difference) and (ii) disproving Conjecture~\ref{conj:upperbound_symdif} by proving that the best upper bound factor we can hope for is $\frac 53$. We complete these results by providing improved upper and lower bounds on the length of transformations in the non-crossing and rotation models defined later. In particular, we close the gap between upper and lower bounds in the case of non-crossing flips.

\paragraph{Improved upper bound.}

The first main result of this paper is to improve the best upper bound of~\cite{bousquet2023note} by a linear factor by proving that the following holds:

\begin{restatable}{theorem}{mainthm}
\label{thm:mainthm}
Let $C$ be a set of $n$ points in convex position.
There exists a flip sequence between any pair of non-crossing spanning trees $T_1$ and $T_2$ of length at most $c \cdot \delta(T_1,T_2)$ with $c = \frac{1}{12}(22 + \sqrt{2})\approx 1.95$.  \\
In particular, there exists a flip sequence of length at most $cn \approx 1.95 n$ between any pair of non-crossing spanning trees.
\end{restatable}

One can note that our result is expressed in terms of the symmetric difference, which is also the case for the upper bound of~\cite{aichholzer2022reconfiguration}\footnote{Some of the partial results obtained in~\cite{aichholzer2022reconfiguration} depend on both $n$ and $\Delta$.}

Our proof technique is completely different from the previous approaches of~\cite{aichholzer2022reconfiguration,bousquet2023note}. Even if these two proofs are very different, they share a common point: their goal is to transform at least one spanning tree into a very rigid structure that does not really take into account the specific structure of both trees. 
On the contrary, our approach depends on the local structure of both trees, which works along the following lines: If two non-crossing trees $T,T'$ contain a common chord, we divide our problem into two sub-problems\footnote{This approach is safe since our upper bound depends on the size of the symmetric difference and not of $n$.}: the ``left" and the ``right" problem where the common chord becomes an edge of the border in both cases. In particular, if we can create a common chord with few modifications, we can apply induction. Unfortunately, this cannot work in general since we may have to modify a lot of edges until we can create a common edge (see e.g. the example of Figure~\ref{fig:hernando}). We prove that we can find a chord $e$ in $T$ and one side of that chord (say ``left") such that $T'$ has not too many endpoints in that side. The difference with the argument above is that the ``not too many" here is not a universal constant but depends on the size of the side. We then prove that, by only modifying a small linear fraction of these edges, we can transform ``left" into what we call a very good side. Informally speaking, ``left" is a very good side if (i) no edge of $T'$ has both endpoint in ``left" (in other words, all the edges with one endpoint in left have the other endpoint in the right part) and, (ii) the number of such edges in $T'$ is equal to the size of ``left". We then prove that, in that case, we can perform flips in order to be sure that both trees agree on the left of $e$ in at most $\frac 53$ times the number of vertices at the left of $e$\footnote{Actually, the size of a side will be defined as the number of non-edges of the border and not simply of vertices in the side which explains why we obtain a bound in terms of the symmetric difference.}.

Our proof is self-contained and is algorithmic. So a flip sequence of length at most $cn$ can be obtained in polynomial time. Moreover it is robust since it can be adapted to improve the best upper bounds for rotations for instance.
Note that we did not try to optimize the constant $c$ to keep the proof as simple as possible. 

\paragraph{Lower bound in terms of the symmetric difference.}

Our second set of results consists in proving stronger lower bounds in terms of the symmetric difference. In particular, we disprove Conjecture~\ref{conj:upperbound_symdif}:

\begin{restatable}{theorem}{lbflips} \label{thm:lb_flips} 
For every $k=0$ mod $3$, there exist two trees $T_k$ and $T_k'$ such that $\delta(T_k,T_k')=k$ and every flip sequence between $T_k$ and $T_k'$ has length at least $\frac 53 k$.
\end{restatable}

 The proof of Theorem~\ref{thm:lb_flips} consists in first providing two spanning trees $T_1,T_1'$ on $8$ vertices for which $|T_1 \setminus T_1'|=3$ and such that the minimal flip sequence between $T_1$ and $T_2$ needs $5$ flips (see Figure~\ref{fig:lower}). 
 One of the reasons of the hardness comes from the fact that, for every common edge $e$ of the border, the endpoint of $e$ that is used to connect this edge to the rest of the tree is different in both trees. This allows us to increase the number of crossings between the trees, and then the length of the flip sequence.  Note that our example is not a counter-example to Conjecture~\ref{conj:upperbound} since the pair of trees contains a lot of common edges.
 
We then prove that if we glue many instances of $(T_1,T_1')$ appropriately, we can obtain a similar example with arbitrarily large value of $k$. 
The idea is as follows. If we assume that there always exists a minimal flip sequence that does not break common edges, the conclusion immediately follows. Unfortunately, this statement, known as the Happy Edge conjecture~\cite{aichholzer2022reconfiguration}, is only known to be true for common edges of the border but not for chords. So we have to prove that it is never interesting to break a common edge which we succeed to do in this particular case (in other words, the Happy Edge Conjecture holds in this case).

We have not found any example for trees $T_1,T_2$ for which a flip sequence of length more than $\frac 53 \delta(T_1,T_2)$ is necessary. We therefore leave the following as an open problem:

\begin{question}
    Let $C$ be a set of points in convex position and $T_1,T_2$ two non-crossing spanning trees on $C$. Does there always exist a flip sequence between $T_1$ and $T_2$ of length at most $\frac 53 \delta(T_1,T_2)$?
\end{question}

\paragraph{Improved lower bounds for the other models.}

Several other types of flips have been introduced in the literature (see e.g.~\cite{nichols2020types} for an overview of the results in the different models). We proved that we can strengthen the best lower bounds in terms of the symmetric difference for two other types of flips: non-crossing flips and rotations. 

Let $T$ be a spanning tree, $e$ be an edge of $T$ and $f$ be an edge such that $(T \cup f) \setminus e$ is non-crossing. We say that the flip is \emph{non-crossing} if $T \cup f$ does not contain any crossings. In other words, we restrict to flips where the new edge does not cross the edge that is deleted. We say that the flip is a \emph{rotation} if $e$ and $f$ share an endpoint. In other words, every flip must rotate an edge around a point in that case. 

Upper and lower bounds in terms of $n$ for the longest minimal possible transformation have already been studied (see~\cite{nichols2020types}). Note that the best known lower bounds for all the models are the same and are given by the construction of Hernando et al. which gives a lower bound (in terms of $n$ of size $\frac 32 n$). We improved the lower bounds in terms of the symmetric difference for both non-crossing flips and rotations.

For non-crossing flips, one can easily remark that, by projecting edges on the border, we can always find a non-crossing flip sequence between any pair of trees of length at most $2\delta(T_1,T_2)$ (see Lemma~\ref{lem:flipborder} for a formal proof). We prove that this bound is tight by giving a pair of trees that reach this bound, which completely closes the gap between lower and upper bounds for non-crossing spanning trees in terms of symmetric difference. Namely we prove that the following holds:

\begin{restatable}{theorem}{lbncflips}\label{thm:lb_ncflips}
For every $k$, there exist two trees $T_k$ and $T_k'$ such that $\delta(T_k,T_k')=k$ and the length of a minimal non-crossing flip sequence between $T_k$ and $T_k'$ has length at least $2k$.
\end{restatable}

We finally consider the rotation model. One can easily remark that there is always a rotation sequence between $T_1$ and $T_2$ of length at most $4\delta(T_1,T_2)$ using projections on the border. Actually one can prove that this $4$ can be improved into a $3$ with a simple clever analysis.
A careful reading of the proof of Theorem~\ref{thm:mainthm} with a slight adaptation actually permits to improve the factor into $(1+c) \approx 2.95$. 

Our last result consists in improving the best lower bound for rotation by showing that the following holds:

\begin{restatable}{theorem}{lbrotations} \label{thm:lb_rotations} 
For every $k=0$ mod $3$, there exist two trees $T_k$ and $T_k'$ such that $\delta(T_k,T_k')=k$ and every rotation sequence between $T_k$ and $T_k'$ has length at least $\frac 73 k$.
\end{restatable}

While the family of trees reaching that bound is similar to the family constructed for flips, the analysis that this family works is much more involved.
We end this part with a last open problem:

\begin{question}
    Let $C$ be a set of points in convex position and $T_1,T_2$ two non-crossing spanning trees on $C$. Does there always exist a rotation sequence from $T_1$ to $T_2$ of length at most $\frac 73 \delta(T_1,T_2)$?
\end{question}


\subsection{Related work}

\paragraph{Flip distance between geometric structures.}

Flips between combinatorial structures have been widely studied in computational geometry and combinatorics. One of the most studied problem, known as the \textsc{Flip Distance} problem, aims at computing the minimum number of flips needed to transform one triangulation into another\footnote{A flip in that case consists in replacing one diagonal of a quadrilateral into the other.} The problem has been proven to be NP-complete when considering $n$ points in non-convex position~\cite{Pilz14,Lubiw15}, and in that case, the flip graph of triangulations of a point set may have diameter $\Theta(n^2)$~\cite{hurtado1996flipping}. When the $n$ points are in convex position, the maximum flip distance between triangulations is linear and equal to $2n-10$ when $n \ge 9$. A first proof for $n$ large was found using hyperbolic geometry~\cite{sleator1986}, while a combinatorial proof for all $n\geq 9$ was only given decades later~\cite{Pournin14}. However, the complexity of the \textsc{Flip Distance} problem is, as far as we know, still an open problem in that case.

 Flip graphs and their diameter for other geometric objects have  been studied, such as non-crossing perfect matchings or rectangulations. For both of these objects there are several natural notions of flips, yielding various flip
graphs. A natural way of defining a flip for perfect matchings is by allowing two edges to be removed and two other edges to be added such that the resulting matching is still non-crossing. 
 When the $n$ points are in convex position and $n$ is even, Hernando, Hurtado and Noy~\cite{hernando2002graphs} showed that the flip graph of non-crossing perfect matchings has diameter $\frac{n}{2}-1$. 
 Houle et al.~\cite{houle2005graphs} gave a result on general point sets when using the notion of flip where $M_1$ is connected to $M_2$ in the flip graph where the symmetric difference of $M_1$ and $M_2$ contains a single non-crossing cycle. They showed that there is a transformation of linear length between any pair of non-crossing matchings, whereas Aichholzer et al.~\cite{aichholzer2009compatible} showed that, if multiple non-crossing cycles are allowed in a flip, then any minimal transformation has length at most $O(\log n)$.

 Ackerman et al.~\cite{ackerman2016flip} considered flips of rectangulations with two elementary flip operations, where one flip changes a horizontal line to a vertical line and vice versa and the other flip is a rotation around a point (by splitting the line segment into two parts). They showed that the maximum flip sequence over all $n$ points is of the order $\Theta(n \log n)$. A natural point set for rectangulations to consider is a diagonal point set, for which Ackerman et al. showed that the flip graph has diameter at most $11n$.
 
\paragraph{Combinatorial Reconfiguration.}
In the last decade, an important line of work has consisted in finding transformations between solutions of a problem such as graph colorings or independent sets (see e.g.~\cite{Nishimura18} for a recent survey). 
Amongst all these works, some of them studied transformations between restricted spanning trees. While we focus in this work on a restriction to the geometric representation of the spanning trees (non-crossing), these works focus on combinatorial properties of the spanning trees such as their maximum degree~\cite{BousquetIKMOSW23} or their number of leaves~\cite{BousquetI0MOSW20}. In these cases, the existence of a transformation is not guaranteed and the goal is to design efficient algorithms determining, given a pair of spanning trees, whether one can transform one into the other. These works focus on the token jumping model which essentially corresponds to flips and very few is known on the token sliding model (which is an analogue of rotations).

As a final remark, spanning trees are, as we already mentioned, a particular case of matroids (called graphic matroids). Other reconfiguration results related to generalizations of matroids have also been studied in the litterature, see e.g.~\cite{KobayashiMS23+}.

\paragraph{Organization of the paper.} After giving some definitions and simple observations in Section~\ref{sec:prelim}, we prove Theorem~\ref{thm:mainthm} in Section~\ref{sec:upper}. In Section~\ref{sec:lower} we prove Theorems~\ref{thm:lb_flips},~\ref{thm:lb_ncflips} and~\ref{thm:lb_rotations}.

\section{Basic definitions and observations}\label{sec:prelim}

    Let $C$ be a set of points in convex position and $T$ be a non-crossing spanning tree on $C$. We say two points of a convex set $C$ are \emph{consecutive} if they appear consecutively on the convex hull of $C$. We say we \emph{perform $e \rightsquigarrow e'$ in $T$} if we perform the flip consisting in removing $e$ and adding $e'$ in $T$.

    Let $A \subseteq C$. We denote by $T[A]$ the induced subgraph of $T$ on $A$, that is the subforest of $T$ with vertex set $A$ where $uv$ is an edge if and only if $uv$ is an edge of $T$. Note that $T[A]$ is a non-crossing forest.
     A \emph{border edge} (for $T$) is an edge between consecutive points. An edge of $T$ which is not a border edge is called a \emph{chord}. A \emph{hole} of $T$ is a pair of consecutive points that is not a border edge. We will say that we \emph{fill a hole} when we apply a flip where the created edge joins the  pair of points of the hole. 
   
    One can remark that, for each chord $e$ of $T$, the line containing $e$ splits the convex hull of $C$ in two non-trivial parts. A \emph{side of a chord $e$} is the subset of points of $C$ contained in one of the two closed half-planes defined by the line containing the two endpoints of $e$ (see Figure~\ref{fig:side} for an illustration).  A \emph{side of $T$} is a side of a chord $e$ for some $e\in T$. We say an edge (or a hole) is \emph{in a side} $A$ if both its endpoints are in $A$. 
    
    In the following, for every side $A$ of a chord, we will denote by $k_A$ the number of holes in $A$, which is also the number of chords of $T$ in $A$. Since $T$ is acyclic, we also have $k_A>0$. Note that each chord $e$ of $T$ defines two non-trivial sides\footnote{A side is trivial if it only contains either two points or all the points.} $A$ and $B$ whose intersection is exactly the endpoints of $e$. Moreover, $T$ has exactly $k_A + k_B$ holes. 
    
    Let $e$ be a chord of $T$ and $A$ be a side of $e$. For every chord $e'$ in $A$, the side of $e'$ (w.r.t. $A$) is the side of $e'$ that is contained in $A$. Note that for every pair of chords $e_1,e_2$ in $A$, the sides of $e_1$ and $e_2$ (w.r.t. $A$) are either disjoint or contained in each other. The chord $e_1$ is \emph{inclusion-wise minimal} if no side of a chord $e'$ in $A$ is included in the side of $e_1$ w.r.t. $A$. By connectivity, we can easily note the following.
    \begin{remark}\label{rmk:min_incl}
        Let $e$ be a chord of $T$ and $A$ be a side of $e$. Let $A'$ be the side of an inclusion-wise minimal chord $e'$ in $A$. Then $k_{A'}=1$. 
    \end{remark}
    
\begin{figure}[hbtp]
        \begin{center}
        \tikzstyle{vertex}=[circle,draw, minimum size=7pt, scale=0.5, inner sep=1pt, fill = black]
        \tikzstyle{fleche}=[->,>=latex]
        \tikzstyle{labell}=[text opacity=1, scale =1.2]
        \begin{tikzpicture}[scale=1.2]

         \begin{scope}
            \clip (0,0) circle (1.5cm);
            \fill[black!12] (90:1.5) to (-90:1.5) to[bend left = 90] (180:1.5) to[bend left = 90] (90:1.5);
            \fill[red!15] (90:1.5) to (-90:1.5) to[bend right = 90] (0:1.5) to[bend right = 90] (90:1.5);
            
        \end{scope}
        \nodearc(a1)(0,0)(90:1.5)[vertex]();
        \nodearc(a2)(0,0)(-90:1.5)[vertex]();
        \nodearc(a3)(0,0)(-30:1.5)[vertex]();
        \nodearc(a4)(0,0)(30:1.5)[vertex]();
        \nodearc(a5)(0,0)(150:1.5)[vertex]();
        \nodearc(a6)(0,0)(-150:1.5)[vertex]();

        \centerarc[thick](0,0)(-90:-30:1.5);
        \centerarc[thick](0,0)(30:90:1.5);
        \centerarc[thick](0,0)(150:210:1.5);

        \draw[thick] (a1) to (a2);
        \draw[thick] (a1) to (a6); 
        
        \nodearc(h4)(0,0)(120:1.8)[labell]($h_1$);
        \nodearc(h5)(0,0)(-120:1.8)[labell]($h_2$);
        \nodearc(h4)(0,0)(0:1.8)[labell]($h_3$);
        \nodearc(h4)(0,0)(0:-0.2)[labell]($e$);

        \nodearc(a1)(0,0)(90:1.8)[labell, scale = 0.8]($v_1$);
        \nodearc(a2)(0,0)(-90:1.8)[labell, scale = 0.8]($v_4$);
        \nodearc(a3)(0,0)(-30:1.8)[labell, scale = 0.8]($v_3$);
        \nodearc(a4)(0,0)(30:1.8)[labell, scale = 0.8]($v_2$);
        \nodearc(a5)(0,0)(150:1.8)[labell, scale = 0.8]($v_6$);
        \nodearc(a6)(0,0)(-150:1.8)[labell, scale = 0.8]($v_5$);

        \end{tikzpicture}
        \end{center}
        \caption{ The side $A$ (in grey) of the chord $e$ is is the subset of vertex $\{ v_1,v_4, v_5, v_6\}$ and the other side $B$ (in red) of $e$ is $\{ v_1,v_4, v_2, v_3\}$. The edges of $T$ in $A$ are the edges $v_5v_6,v_1v_5$ and $v_1v_4$. The holes $h_1$ and $h_2$ of $T$ are in $A$ and $h_3$ is in $B$. So we have $k_A=2$ and $k_B = 1$.}
        \label{fig:side}
\end{figure}
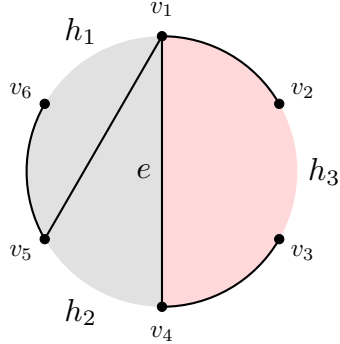

    Let $A$ be a side of a chord $e$ of $T$. 
    We say that a point $v$ is \emph{inside $A$ for} a pair $vw$ if $v \in A$ and either $v$ is not an endpoint of $e$ or both $v$ and $w$ are in $A$ (see Figure~\ref{fig:inside} for an illustration).
    Note that the fact that $v$ is inside $A$ depends on the pair $vw$. Note moreover that a point can be inside $A$ for several pairs of points, and a point can be inside $A$ for some pairs but not inside $A$ for other pairs. 
    We define the \emph{degree} of a side $A$ in a tree $T'$ as the number of endpoints (counted with multiplicity) that are inside $A$ for some chord of $T'$.

\usetikzlibrary{intersections}

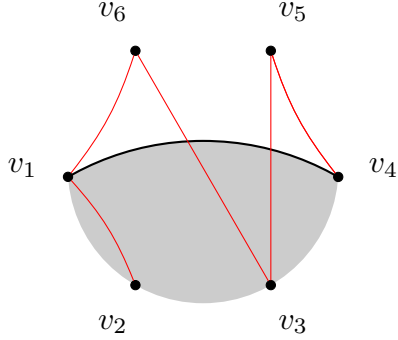
\begin{figure}[hbtp]
        \begin{center}
        \tikzstyle{vertex}=[circle,draw, minimum size=7pt, scale=0.5, inner sep=1pt, fill = black]
        \tikzstyle{fleche}=[->,>=latex]
        \tikzstyle{labell}=[text opacity=1, scale =1.2]
        \begin{tikzpicture}[scale=1.2]

        \def\firstcircle{(0,0) circle (1.5cm)}
        \def\secondcircle{(0,-2.7) circle (3cm)}
        \begin{scope}[shift={(3cm,-5cm)}]
            \draw [white, name path = convex] \firstcircle node[black] {};
            \begin{scope}
                \clip \firstcircle;
                \fill[fill opacity = 0.2] \secondcircle;
                \draw [name path = edge1, thick] \secondcircle;
            \fill[fill opacity = 0, name intersections={of=convex and edge1,total=\t}]
                \foreach \s in {1,...,\t}{(intersection-\s) node {}};
            \end{scope}
            \node[vertex] (i1) at (intersection-1) {} ;
            \node[vertex] (i2) at (intersection-2) {} ;
            \node[vertex] (a1) at (-120:1.5cm) {};
            \node[vertex] (a2) at (-60:1.5cm) {};
            \node[vertex] (a3) at (60:1.5cm) {};
            \node[vertex] (a4) at (120:1.5cm) {};
            \begin{scope}
                \clip \firstcircle;
                \draw (i1) to [bend left= 10] (a1) [red];
                \draw (i2) to [bend left = 10] (a3) [red];
                \draw (i2) to [bend left = 10] (a3) [red];
                \draw (a2) to (a4) [red];
                \draw (i1) to [bend right = 10] (a4) [red];
                \draw (a2) to (a3) [red];
            \end{scope}
            \node[labell] (l1) at (180:2cm) {$v_1$};
            \node[labell] (l2) at (-120:2cm) {$v_2$};
            \node[labell] (l3) at (-60:2cm) {$v_3$};
            \node[labell] (l4) at (0:2cm) {$v_4$};
            \node[labell] (l5) at (60:2cm) {$v_5$};
            \node[labell] (l6) at (120:2cm) {$v_6$};
        \end{scope}
        \end{tikzpicture}
        \end{center}
        \caption{The side $A$ of the edge $v_1v_4$ highlighted in grey contains $v_1,v_2,v_3$ and $v_4$. The degree of $A$ in the red tree is equal to $4$: $v_1$ and $v_2$ are inside $A$ for $v_1v_2$; and $v_3$ is twice inside $A$, once for $v_3v_5$ and once for $v_3v_6$.
        The vertex $v_1$ is outside $A$ for $v_1v_6$ (even if it is inside for $v_1v_2$), and so is $v_4$ for $v_4v_5$. }
        \label{fig:inside}
\end{figure}

The following lemma appeared in~\cite{bousquet2023note}. We give the proof for completeness.

\begin{lemma}\label{lem:flipborder}
    Let $T$ be a tree and $e$ be a border edge. Then there exists a non-crossing flip that adds $e$ in $T$ without removing any border edge of $T$ (except if $T$ only contains border edges).
\end{lemma}
\begin{proof}
    Adding $e$ to $T$ does not create any crossing, since $e$ is a border edge. 
    Moreover, the unique cycle in $T \cup \{ e \}$ must contain at least one chord $e'$, since otherwise $T \cup \{ e \}$ is precisely the convex hull of $C$. Adding $e$ and removing $e'$ is a non-crossing flip as claimed.
\end{proof}

\section{Upper bounds}\label{sec:upper}
This section aims at proving Theorem~\ref{thm:mainthm}.
\mainthm*

We say that two trees $T_I$ and $T_F$ on a convex point set $C$ form a \emph{minimal counterexample}  if the pair $(T_I,T_F)$ is a counterexample to Theorem~\ref{thm:mainthm}, and for every pair of trees $T'_I$ and $T'_F$, which are either defined on the same set of points and $\delta(T'_I, T'_F) < \delta(T_I, T_F)$ or on a smaller set of points, Theorem~\ref{thm:mainthm} holds.

Before giving all the details of the proof, let us first explain the main steps of the proof (see Figure~\ref{fig:upper} for an illustration).
First, we prove in Section~\ref{sec:bascounter} that a minimal counterexample $T_I$ and $T_F$ is non-trivially reducible, i.e. $T_I$ and $T_F$ have the same border edges and no common chord. 
In Section~\ref{sec:verygood}, we will observe that, in some sides (later called very good), we can match the $k_A$ chords in the side using at most $\frac{5}{3}k_A$ flips in total.
However, a very good side does not necessarily exist in a minimal counterexample. 
Thus, we will give tools in Section~\ref{sec:comb} to obtain a very good side from another special type of side (which we can also find in a minimal counterexample) without using too many flips.

\begin{figure}[hbtp]
        \begin{center}
        \tikzstyle{vertex}=[circle,draw, minimum size=7pt, scale=0.5, inner sep=1pt, fill = black]
        \tikzstyle{fleche}=[->,>=latex]
        \tikzstyle{labell}=[text opacity=1, scale =1.2]
        \begin{tikzpicture}[scale=1]

            \def\firstarc{($(0,1)+({3*cos(180)},{3*sin(180)})$) arc (-180:0:3)}
            
            \nodearc(a1)(0,1)(-140:3)[vertex]();
            \nodearc(a2)(0,1)(-40:3)[vertex]();
            \nodearc(a3)(0,1)(-60:3)[vertex]();
            \nodearc(a4)(0,1)(-80:3)[vertex]();
            \nodearc(a5)(0,1)(-100:3)[vertex]();
            \nodearc(a6)(0,1)(-120:3)[vertex]();

            \begin{scope}
            \clip \firstarc;
            \fill[black!10] (a1) to[bend left = 70] (a2) to[bend left = 90] (a1);
            \draw[thick]  (a1) to[bend left = 70] (a2) to[bend left = 90] (a1); 
            \end{scope}

            \centerarc[thick, dashed](0,1)(-140:-150:3);
            \centerarc[thick, dashed](0,1)(-40:-30:3);
            
            \centerarc[thick](0,1)(-140:-120:3);
            \centerarc[thick](0,1)(-80:-100:3);
            \draw[thick] (a5) to [bend left=60] (a3);
            \draw[thick] (a1) to [bend left=50] (a5);
            
            \draw[thick, red] (a1) to [bend left=20] (a6);
            \draw[thick, red] (a5) to [bend left=20] (a4);
            \draw[thick, red] (a6) to [bend left=60] (a4);

              \nodearc(a1)(0,1)(-140:3)[vertex]();
            \nodearc(a2)(0,1)(-40:3)[vertex]();
            \nodearc(a3)(0,1)(-60:3)[vertex]();
            \nodearc(a4)(0,1)(-80:3)[vertex]();
            \nodearc(a5)(0,1)(-100:3)[vertex]();
            \nodearc(a6)(0,1)(-120:3)[vertex]();

            \draw[thick, red] (a3) to  ({3*cos(-55)},0.5);
            \draw[thick, red] (a3) to  ({3*cos(-65)},0.5);
            \draw[thick, red] (a4) to  ({3*cos(-80)},0.5);
            \draw[thick, red] (a6) to  ({3*cos(-120)},0.5);
            \draw[thick, red, dashed] (a2) to  ({3*cos(-37)},-0.4);
            \draw[thick, red, dashed] (a1) to  ({3*cos(-143)},-0.4);

            \node (h0) at (0, 1) [draw, scale=1]{$\tau$-extremal side};

        \tikzset{xshift=3cm}

            \node (h0) at (1, -0.5) [labell, scale=0.8]{Lemma~\ref{lem:bad}};
            \draw[->, very thick] (0,-1) to (2,-1);
            \node (h0) at (1, -1.5) [labell, scale=0.8]{fill bad holes};
            
        \tikzset{xshift=5cm}

            \def\firstarc{($(0,1)+({3*cos(180)},{3*sin(180)})$) arc (-180:0:3)}
            
            \nodearc(a1)(0,1)(-140:3)[vertex]();
            \nodearc(a2)(0,1)(-40:3)[vertex]();
            \nodearc(a3)(0,1)(-60:3)[vertex]();
            \nodearc(a4)(0,1)(-80:3)[vertex]();
            \nodearc(a5)(0,1)(-100:3)[vertex]();
            \nodearc(a6)(0,1)(-120:3)[vertex]();

            \begin{scope}
            \clip \firstarc;
            \fill[black!10] (a1) to[bend left = 70] (a2) to[bend left = 90] (a1);
            \draw[thick]  (a1) to[bend left = 70] (a2) to[bend left = 90] (a1); 
            \end{scope}

            \centerarc[thick, dashed](0,1)(-140:-150:3);
            \centerarc[thick, dashed](0,1)(-40:-30:3);
            
            \centerarc[thick](0,1)(-140:-80:3);
            \centerarc[thick](0,1)(-80:-100:3);
            \draw[thick] (a5) to [bend left=60] (a3);
            \draw[thick, red] (a1) to [bend left=20] (a6);
            \draw[thick, red] (a5) to [bend left=20] (a4);
            \draw[thick, red] (a6) to [bend left=20] (a5);

              \nodearc(a1)(0,1)(-140:3)[vertex]();
            \nodearc(a2)(0,1)(-40:3)[vertex]();
            \nodearc(a3)(0,1)(-60:3)[vertex]();
            \nodearc(a4)(0,1)(-80:3)[vertex]();
            \nodearc(a5)(0,1)(-100:3)[vertex]();
            \nodearc(a6)(0,1)(-120:3)[vertex]();

            \draw[thick, red] (a3) to  ({3*cos(-55)},0.5);
            \draw[thick, red] (a3) to  ({3*cos(-65)},0.5);
            \draw[thick, red] (a4) to  ({3*cos(-80)},0.5);
            \draw[thick, red] (a6) to  ({3*cos(-120)},0.5);
            \draw[thick, red, dashed] (a2) to  ({3*cos(-37)},-0.4);
            \draw[thick, red, dashed] (a1) to  ({3*cos(-143)},-0.4);

            \node (h0) at (0, 1) [draw, scale = 1]{good side};

            \node (h0) at (1.1, -2.8) [labell, scale=0.8]{Lemma~\ref{lem:good}};
            \draw[->, very thick] (0,-2.2) to (0,-3.4);
            \node (h0) at (-1.3, -2.6) [labell, scale=0.8]{remove extra};
            \node (h0) at (-1.3, -3.0) [labell, scale=0.8]{crossing chords};

        \tikzset{yshift=-4cm}

            \def\firstarc{($(0,1)+({3*cos(180)},{3*sin(180)})$) arc (-180:0:3)}
            
            \nodearc(a1)(0,1)(-140:3)[vertex]();
            \nodearc(a2)(0,1)(-40:3)[vertex]();
            \nodearc(a3)(0,1)(-60:3)[vertex]();
            \nodearc(a4)(0,1)(-80:3)[vertex]();
            \nodearc(a5)(0,1)(-100:3)[vertex]();
            \nodearc(a6)(0,1)(-120:3)[vertex]();

            \begin{scope}
            \clip \firstarc;
            \fill[black!10] (a1) to[bend left = 70] (a2) to[bend left = 90] (a1);
            \draw[thick]  (a1) to[bend left = 70] (a2) to[bend left = 90] (a1); 
            \end{scope}

            \centerarc[thick, dashed](0,1)(-140:-150:3);
            \centerarc[thick, dashed](0,1)(-40:-30:3);
            
            \centerarc[thick](0,1)(-140:-80:3);
            \centerarc[thick](0,1)(-80:-100:3);
            \draw[thick] (a5) to [bend left=60] (a3);
            \draw[thick, red] (a1) to [bend left=20] (a6);
            \draw[thick, red] (a5) to [bend left=20] (a4);
            \draw[thick, red] (a6) to [bend left=20] (a5);

              \nodearc(a1)(0,1)(-140:3)[vertex]();
            \nodearc(a2)(0,1)(-40:3)[vertex]();
            \nodearc(a3)(0,1)(-60:3)[vertex]();
            \nodearc(a4)(0,1)(-80:3)[vertex]();
            \nodearc(a5)(0,1)(-100:3)[vertex]();
            \nodearc(a6)(0,1)(-120:3)[vertex]();

            \draw[thick, red] (a3) to  ({3*cos(-55)},0.5);
            \draw[thick, red] (a4) to  ({3*cos(-80)},0.5);
            \draw[thick, red, dashed] (a2) to  ({3*cos(-37)},-0.4);
            \draw[thick, red, dashed] (a1) to  ({3*cos(-143)},-0.4);   
            
            \node (h0) at (0,-2.5) [draw, scale=1]{very good side};

            \tikzset{xshift=-5cm}

            \node (h0) at (1, -0.5) [labell, scale=0.8]{Lemma~\ref{lem:very_good}};
            \draw[->, very thick] (2,-1) to (0,-1);
            \node (h0) at (1, -1.5) [labell, scale=0.8]{match remaining};
            \node (h0) at (1, -1.9) [labell, scale=0.8]{chords};
            
            \tikzset{xshift=-3cm}

            \def\firstarc{($(0,1)+({3*cos(180)},{3*sin(180)})$) arc (-180:0:3)}
            
            \nodearc(a1)(0,1)(-140:3)[vertex]();
            \nodearc(a2)(0,1)(-40:3)[vertex]();
            \nodearc(a3)(0,1)(-60:3)[vertex]();
            \nodearc(a4)(0,1)(-80:3)[vertex]();
            \nodearc(a5)(0,1)(-100:3)[vertex]();
            \nodearc(a6)(0,1)(-120:3)[vertex]();

            \begin{scope}
            \clip \firstarc;
            \fill[black!10] (a1) to[bend left = 70] (a2) to[bend left = 90] (a1);
            \draw[thick]  (a1) to[bend left = 70] (a2); 
            \draw[thick, red]  (a1) to[bend left = 50] (a2); 
            \end{scope}

            \centerarc[thick, dashed](0,1)(-140:-150:3);
            \centerarc[thick, dashed](0,1)(-40:-30:3);
            
            \centerarc[thick](0,1)(-140:-80:3);
            \centerarc[thick](0,1)(-80:-100:3);
            \draw[thick] (a5) to [bend left=60] (a3);
            \draw[thick, red] (a5) to [bend left=40] (a3);
            \draw[thick, red] (a1) to [bend left=20] (a6);
            \draw[thick, red] (a5) to [bend left=20] (a4);
            \draw[thick, red] (a6) to [bend left=20] (a5);

              \nodearc(a1)(0,1)(-140:3)[vertex]();
            \nodearc(a2)(0,1)(-40:3)[vertex]();
            \nodearc(a3)(0,1)(-60:3)[vertex]();
            \nodearc(a4)(0,1)(-80:3)[vertex]();
            \nodearc(a5)(0,1)(-100:3)[vertex]();
            \nodearc(a6)(0,1)(-120:3)[vertex]();

            \draw[thick, red, dashed] (a2) to  ({3*cos(-37)},-0.4);
            \draw[thick, red, dashed] (a1) to  ({3*cos(-143)},-0.4); 
            
            \node (h0) at (0,-2.5) [draw, scale=1]{matched side};

        \end{tikzpicture}
        \end{center}
        \caption{The main steps in the proof of Theorem~\ref{thm:mainthm}. The goal is to match the chords in a side using few steps. (We define $\tau$-extremal, bad, good and very good later).}
        \label{fig:upper}
\end{figure}
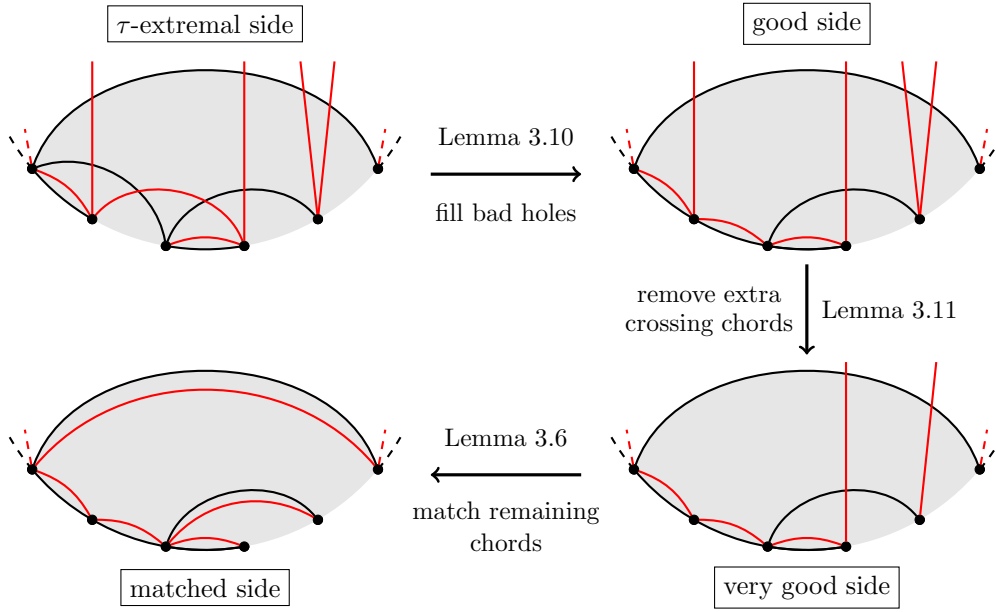

\subsection{Basic properties of a minimal counterexample}\label{sec:bascounter}

First, we prove the base case.

\begin{lemma}
Let $T_I$ and $T_F$ be a minimal counterexample.
Then both $T_I$ and $T_F$ have a chord.
\end{lemma}
\begin{proof}
Assume by contradiction that $T_I$ contains only border edges. By Lemma~\ref{lem:flipborder}, we can iteratively add border edges to $T_F$ by decreasing the symmetric difference at each step until the two trees are the same. So there is a flip sequence between $T_I$ and $T_F$ of length at most $\delta(T_I,T_F)$, a contradiction. The other case holds by symmetry.
\end{proof}

\begin{lemma}\label{lem:common_chord}
Let $T_I$ and $T_F$ be a minimal counterexample.
Then $T_I$ and $T_F$ do not have common chords.
\end{lemma}
\begin{proof}
    Assume by contradiction that $T_I$ and $T_F$ have a common chord $e$.
    Let $A $ and $B$ be the two sides of $e$.  
    Since $e$ is a chord, $|A|$ and $|B|$ are smaller than $|C|$.
    Since $e$ belongs to both $T_I$ and $T_F$, $T_I[A]$ and $T_F[A]$ is a pair of non-crossing spanning trees of $A$. 
    By minimality of the counterexample, there exists a flip sequence between $T_I[A]$ and $T_F[A]$ of length at most $c \cdot \delta(T_I[A], T_F[A])$. 
    Note that this sequence of flips can be performed for whole set of points $C$ starting from $T_I$ while keeping connectivity. Let $T_I'$ be the tree obtained from $T_I$ after applying this flip sequence. 
    Note that $T_I'$ and $T_F$ agree on $A$. 
    
    The same argument on $T_I'[B]$ and $T_F[B]$ ensures that there also exists a flip sequence between $T_I'[B]$ and $T_F[B]$ of length at most $c \cdot \delta(T_I'[B], T_F[B])$. 
    Hence, there exists a flip sequence sequence between $T_I$ and $T_F$ of length at most $c \cdot \delta(T_I[A], T_F[A]) +  c \cdot \delta(T_I'[B], T_F[B]) = c \cdot \delta(T_I, T_F)$ (the equality holds since $e$ is an edge of all the trees $T_I[A], T_F[A]$, $T_I[B]$ and $T_F[B]$ and then does not belong to the symmetric difference of any pair), a contradiction.
\end{proof}

\begin{lemma}\label{lem:common_border} 
Let $T_I$ and $T_F$ be a minimal counterexample.
    Every border edge of $T_I$ is a border edge of~$T_F$.
\end{lemma}

\begin{proof}
    Assume by contradiction that there is a border edge $e$ of $T_I$ that is not an edge of $T_F$. 
    Since $T_F$ has at least one chord by Lemma~\ref{lem:flipborder}, we can perform $e^* \rightsquigarrow e$ in $T_F$, where $e^*$ is a chord of $T_F$.
    By Lemma~\ref{lem:common_chord}, $e^*$ is not contained in $T_I$.
    Let $T_F^*$ be the non-crossing spanning tree obtained from $T_F$ after flipping $e^*$ into $e$. 
    By minimality of $T_I$ and $T_F$, there exists a flip sequence between $T_I$ and $T_F^*$ of length at most $c \cdot \delta(T_I, T^*_F) = c(\delta(T_I, T_F) - 1)$.
    Thus, there exists a flip sequence between $T_I$ and $T_F$ of length at most $c \cdot \delta(T_I, T^*_F) + 1 < c \cdot \delta(T_I, T_F) $, a contradiction.
\end{proof}

We say that two trees $T$ and $T'$ form a \emph{nice pair} of trees if the two trees have no common chord and have the same border edges. 
Note that for a nice pair of trees, every pair of consecutive points is either a common hole or a common border edge. 
Thus, for a nice pair of trees $(T_1, T_2)$, we will refer to a hole of $T_1$ or $T_2$ simply as a hole.
Lemmas~\ref{lem:common_chord} and~\ref{lem:common_border} ensure that the following holds:

\begin{corollary}\label{cor:nice}
    A minimal counterexample is a nice pair of trees. 
\end{corollary}

A \emph{face} $f$ of a tree $T$ is a face, different from the outer face, of the plane graph obtained by filling the holes of $T$ with edges. 
Note that, since $T$ is connected, every face contains exactly one hole on its boundary and every hole is on the boundary of exactly one face. 
Thus, there is a bijection between holes and faces of a tree. 
The face \emph{containing} a hole $h$ of a tree $T$ is the face $f$ such that $h$ belongs to the boundary of $f$. We say that the hole $h$ is \emph{contained in} the face $f$ in $T$. 

\begin{figure}[hbtp]
        \begin{center}
        \tikzstyle{vertex}=[circle,draw, minimum size=7pt, scale=0.5, inner sep=1pt, fill = black]
        \tikzstyle{fleche}=[->,>=latex]
        \tikzstyle{labell}=[text opacity=1, scale =1.2]
        \begin{tikzpicture}[scale=1.2]

         \begin{scope}
            \clip (0,0) circle (1.5cm);
            \fill[black!10] (90:1.5) to (-150:1.5) to[bend left = 90] (90:1.5);
            \draw[pattern=hatch, hatch size=8, hatch angle = 80, pattern color=orange!70] (90:1.5) to (-150:1.5) to[bend right = 90] (-90:1.5) to (-90:1.5);
            \fill[red!15] (90:1.5) to (-90:1.5) to[bend right = 90] (0:1.5) to[bend right = 90] (90:1.5);
            
        \end{scope}
        \nodearc(a1)(0,0)(90:1.5)[vertex]();
        \nodearc(a2)(0,0)(-90:1.5)[vertex]();
        \nodearc(a3)(0,0)(-30:1.5)[vertex]();
        \nodearc(a4)(0,0)(30:1.5)[vertex]();
        \nodearc(a5)(0,0)(150:1.5)[vertex]();
        \nodearc(a6)(0,0)(-150:1.5)[vertex]();

        \centerarc[thick](0,0)(-90:-30:1.5);
        \centerarc[thick](0,0)(30:90:1.5);
        \centerarc[thick](0,0)(150:210:1.5);

        \draw[thick] (a1) to (a2);
        \draw[thick] (a1) to (a6); 
        
        \nodearc(h4)(0,0)(120:1.8)[labell]($h_1$);
        \nodearc(h5)(0,0)(-120:1.8)[labell]($h_2$);
        \nodearc(h4)(0,0)(0:1.8)[labell]($h_3$);
        \nodearc(h4)(0,0)(150:1.1)[labell]($f_1$);
        \nodearc(h5)(0,0)(-150:0.6)[labell, orange!80!black]($f_2$);
        \nodearc(h4)(0,0)(0:0.75)[labell, red]($f_3$);

        \end{tikzpicture}
        \end{center}
        \caption{The tree $T$ has three faces $f_1$, $f_2$ and $f_3$. The face $f_i$ contains the hole $h_i$ in $T$.}
        \label{fig:face}
\end{figure}
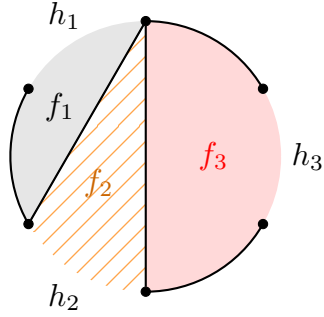

\begin{observation}\label{ob:bound_hole}
    For every hole $h$, we can fill $h$ in $T$ by flipping any chord on the boundary of the face of $T$ containing $h$. 
\end{observation}

\subsection{Very good sides}\label{sec:verygood}

Let $T$ and $T'$ be a pair of trees.
A \emph{good} side $A$ of $T$ with respect to $T'$ is a side of $T$ containing no chord of $T'$ (see Figure~\ref{fig:good} for an illustration). 
A good side $A$ is \emph{very good} (w.r.t. $T'$) if the degree of $A$ in $T'$ is at most $k_A$. (Recall that $k_A$ denotes the number of holes in $A$). Remark that, for a side $A$ of an edge $e$ of $T$, if $A$ is a good side of $T$ w.r.t. a tree $T'$, the degree of $A$ in $T'$ is equal to the number of chords of $T'$ crossing $e$.

\begin{figure}[hbtp]
        \begin{center}
        \tikzstyle{vertex}=[circle,draw, minimum size=7pt, scale=0.5, inner sep=1pt, fill = black]
        \tikzstyle{fleche}=[->,>=latex]
        \tikzstyle{labell}=[text opacity=1, scale =1.2]
        \begin{tikzpicture}[scale=1]

            \def\firstarc{($(0,1)+({3*cos(180)},{3*sin(180)})$) arc (-180:0:3)}
            
            \nodearc(a1)(0,1)(-140:3)[vertex]();
            \nodearc(a2)(0,1)(-40:3)[vertex]();
            \nodearc(a3)(0,1)(-65:3)[vertex]();
            \nodearc(a4)(0,1)(-90:3)[vertex]();
            \nodearc(a5)(0,1)(-115:3)[vertex]();

            \begin{scope}
            \clip \firstarc;
            \fill[black!10] (a1) to[bend left = 70] (a2) to[bend left = 90] (a1); 
            \draw[thick]  (a1) to[bend left = 70] (a2) to[bend left = 90] (a1); 
            \end{scope}

            \centerarc[thick, dashed](0,1)(-140:-150:3);
            \centerarc[thick, dashed](0,1)(-40:-30:3);
            
            \centerarc[thick](0,1)(-140:-115:3);
            \centerarc[thick](0,1)(-90:-65:3);
            \draw[thick] (a1) to [bend left=50] (a4);
            
            \draw[thick, red] (a1) to [bend left=20] (a5);
            \draw[thick, red] (a4) to [bend left=20] (a3);

            \nodearc(a1)(0,1)(-140:3)[vertex]();
            \nodearc(a2)(0,1)(-40:3)[vertex]();
            \nodearc(a3)(0,1)(-65:3)[vertex]();
            \nodearc(a4)(0,1)(-90:3)[vertex]();
            \nodearc(a5)(0,1)(-115:3)[vertex]();

            \draw[thick, red] (a3) to  ({3*cos(-65)},0.4);
            \draw[thick, red] (a5) to  ({3*cos(-110)},0.4);
            \draw[thick, red] (a5) to  ({3*cos(-120)},0.4);
            \draw[thick, red, dashed] (a2) to  ({3*cos(-37)},-0.4);
            \draw[thick, red, dashed] (a1) to  ({3*cos(-143)},-0.4);

            \node (h3) at ({3*cos(90)},0.15) [labell]{$e$};
            \nodearc(a4)(0,1)(-102.5:3.3)[labell]($h$);
            \nodearc(a5)(0,1)(-52.5:3.3)[labell]($h'$);

        \end{tikzpicture}
        \end{center}
        \caption{Let $T_1$ be the black tree and $T_2$ the red tree. The side $A$ (in grey) of $e$ is a good side of $T_1$ w.r.t. $T_2$ since there is no chord of $T_2$ inside $A$, but $A$ is not very good w.r.t $T_2$ since the degree of $A$ in $T_2$ is $3 > k_A = 2$.} 
        \label{fig:good}
\end{figure}
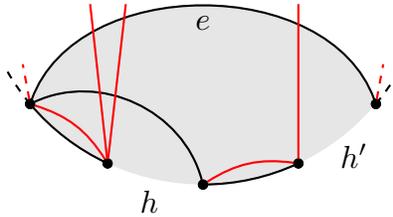

The goal of Section~\ref{sec:verygood} is to prove that the following holds:

\begin{restatable}{lemma}{verygood} \label{lem:very_good}
 Let $T_1$ and $T_2$ be a nice pair of trees, $e$ be a chord of $T_1$, and $A$ be a very good side of $e$ (w.r.t. $T_2$). Then, we can match $k_A$ pairs of chords of $T_1$ and $T_2$ using at most $\frac{5}{3} k_A$ flips in total. 
\end{restatable}

Before proving this lemma, let us introduce some definitions and investigate the structure of very good sides (see Figure~\ref{fig:borderpath} for an illustration of these definitions).
A \emph{border path} of a tree $T$ is a maximal path (possibly reduced to a single point) of consecutive border edges of $T_1$. 
A border path of a tree $T$ and a hole $h$ are \emph{incident} if they share a common point.
Note that, for a tree $T$ that is not reduced to a border path, a border path of $T$ is incident to exactly two distinct holes of $T$, and a hole of $T$ is incident to exactly two distinct border paths of $T$.
Note that since trees forming a nice pair admit the same border edges, they also have the same border paths. 
Thus, for a nice pair of trees $T_1$, $T_2$, we will refer to a border path of $T_1$ or $T_2$ simply as a border path. Let us first remark that the following holds:

\begin{observation}\label{ob:path}
    Let $T$ be a tree and $h$ be a hole of $T$ such that $P$ is a border path of $T$ incident to $h$. If $T$ has a chord, then the face of $T$ containing $h$ has a chord on its boundary that has an endpoint in $P$.
\end{observation}

An \emph{internal border path} $P$ of a side $A$ of a chord $e$ is a border path such that all the points of $P$ are in $A$ and that does not contain an endpoint of $e$ (see Figure~\ref{fig:borderpath} for an illustration).

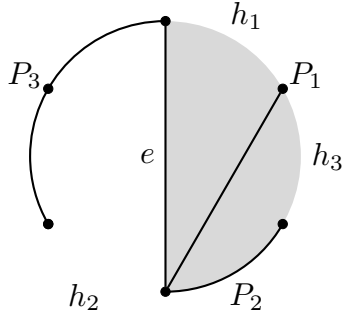
\begin{figure}[hbtp]
        \begin{center}
        \tikzstyle{vertex}=[circle,draw, minimum size=7pt, scale=0.5, inner sep=1pt, fill = black]
        \tikzstyle{fleche}=[->,>=latex]
        \tikzstyle{labell}=[text opacity=1, scale =1.2]
        \begin{tikzpicture}[scale=1.2]

         \begin{scope}
            \clip (0,0) circle (1.5cm);
            \fill[black!15] (90:1.5) to (-90:1.5) to[bend right = 90] (0:1.5) to[bend right = 90] (90:1.5);
        \end{scope}
        \nodearc(a1)(0,0)(90:1.5)[vertex]();
        \nodearc(a2)(0,0)(-90:1.5)[vertex]();
        \nodearc(a3)(0,0)(-30:1.5)[vertex]();
        \nodearc(a4)(0,0)(30:1.5)[vertex]();
        \nodearc(a5)(0,0)(150:1.5)[vertex]();
        \nodearc(a6)(0,0)(-150:1.5)[vertex]();
        \nodearc(a6)(0,0)(-150:1.5)[vertex]();

        \centerarc[thick](0,0)(-90:-30:1.5);
        \centerarc[thick](0,0)(90:210:1.5);

        \draw[thick] (a1) to (a2);
        \draw[thick] (a4) to (a2); 
        
        \nodearc(h4)(0,0)(60:1.8)[labell]($h_1$);
        \nodearc(h5)(0,0)(-120:1.8)[labell]($h_2$);
        \nodearc(h4)(0,0)(0:1.8)[labell]($h_3$);
        \nodearc(h4)(0,0)(30:1.8)[labell]($P_1$);
        \nodearc(h5)(0,0)(-60:1.8)[labell]($P_2$);
        \nodearc(h4)(0,0)(150:1.8)[labell]($P_3$);
        \nodearc(h4)(0,0)(-180:0.2)[labell]($e$);

        \end{tikzpicture}
        \end{center}
        \caption{The tree has three border paths $P_1$, $P_2$ and $P_3$. The border path $P_1$, reduced to a single vertex, is incident to two holes $h_1$ and $h_3$. The border path $P_1$ is an internal border path of the side $A$ (in grey) of $e$ while $P_2$ (containing an endpoint of $e$) and $P_3$ (not in $A$) are not.}
        \label{fig:borderpath}
\end{figure}

\begin{lemma}\label{lem:intborderpath}
    Let $T_1$ and $T_2$ be a nice pair of trees, $e$ be a chord of $T_1$, and $A$ be a very good side of $e$ w.r.t. $T_2$. Then: 
    \begin{itemize}
        \item[(i)] Every internal border path is incident to $1$ or $2$ chords of $T_2$,
        \item[(ii)] Every border path containing an endpoint of $e$ is incident to $0$ or $1$ chord of $T_2$ in $A$,
        \item[(iii)] There is at most one border path reaching the maximum value in (i) and (ii). 
        When it exists, we call this path the \emph{extra path} of $A$. 
    \end{itemize}
\end{lemma}

\begin{proof}
Let us denote by $h_1,\ldots,h_{k_A}$ the holes of $A$ and let us denote by $P_i$ the internal path between $h_i$ and $h_{i+1}$ for every $i < k_A$ (note that there are $k_A-1$ such paths). 
By Observation~\ref{ob:path} applied to $h_i,P_i$ for every $i < k_A$, there is at least one chord of $T_2$ incident to $P_i$. 

So the degree of $A$ in $T_2$ is at least $k_A - 1$. 
    Since the degree of $A$ in $T_2$ is at most $k_A$ (by definition of very good side), there is at most one extra chord of $T_2$ with an endpoint in $A$.
    If it exists, this extra chord is either incident with a path that contains an endpoint of $e$, or an internal border path of $A$. 
\end{proof}

Lemma~\ref{lem:intborderpath} gives us a lot of information on how the tree $T_2$ behaves in a very good side w.r.t $T_2$. Using this information, we can now prove the main result of this section:

\begin{proof}[Proof of Lemma~\ref{lem:very_good}]
Consider a side $A' \subseteq A$ of an inclusion-wise minimal chord $e'$ of $T_1$ w.r.t. $A$. Note that we might have $A=A'$. Let $d$ be the degree of $A'$ in $T_2$. 
By Remark~\ref{rmk:min_incl}, $A'$ contains exactly one hole $h$. Since there is no chord of $T_1$ in $A'$, the two border paths $P_1$ and $P_2$ incident to $h$ contain all the points in $A'$. 
By Lemma~\ref{lem:intborderpath}, these paths contain altogether at most three endpoints of a chord of $T_2$ in $A'$ since $A$ is very good, hence $d \leq 3$. 
The core of the proof consists in proving the following result:

\begin{claim}\label{cl:very_good}
Let $d\leq 3$. In at most $ \frac{5}{3} \max (1, d)$ flips, we can obtain from $T_1$ and $T_2$ a very good pair of trees $T'_1,T'_2$ that agree on $A'$ by filling $\max(0,d-1)$ holes of $A$ in both $T_1$ and $T_2$ and creating the chord $e'$ in $T_2$, see Figure~\ref{fig:verygood}.   Moreover, if $k_A > d$, then we do not flip $e$ and the number of edges in $T'_2$ crossing $e$ decreased by $\max(1,d)$.

 
    \end{claim}

    \input{figures/verygood}

    Before proving Claim~\ref{cl:very_good}, let us show how it can be used to complete the proof of Lemma~\ref{lem:very_good}. Recall that $A$ is very good hence $d \leq k_A$. If $d=k_A$, then we can fill every hole of $A$ that is different from $h$, and then add $e'$ to $T_2$ using at most $\frac 53 d = \frac 53 k_A$ flips by Claim~\ref{cl:very_good}. We may thus assume that $d < k_A$.

    We proceed by induction on $k_A$. For the base case $k_A=1$, we have $d=0$. In that case, $A=A'$ and $e=e'$, and we can add $e$ to $T_2$ using at most $\lfloor \frac 53 \rfloor = 1$ flip by Claim~\ref{cl:very_good}.
    
    Assume now that $k_A>1$. Let $T_1^*$ and $T_2^*$ be the two trees obtained from $T_1$ and $T_2$ after applying Claim~\ref{cl:very_good}.
    Note that even if $e$ was not flipped, we cannot directly apply induction since $T_1^*$ and $T_2^*$ do not form a nice pair ($e'$ is a common chord) and moreover, $e$ is not very good anymore (because of $e'$).
    
    However, since $T_1$ and $T_2$ form a nice pair, and we only added holes and $e'$ to obtain $T_1^*, T_2^*$, $e'$ is the only common chord of the two trees $T_1^*, T_2^*$ and they have the same border edges. Let $C'$ be the convex set of points obtained from $C$ by removing the points of $A'$ that are not endpoints of $e'$.

    Now observe that the trees $T_1^*[C'] $ and $T_2^*[C']$ form a nice pair where $e'$ is a common border edge. We claim that $e$ is now very good for $T_2^*[C']$. Let $A^*$ be the side of $e$ in $T_1^*[C']$ that is contained in $A$. Note $A^*$ does not contain any chord of $T_2^*[C']$, hence $A^*$ is good w.r.t. $T_2^*[C']$. 
    
    To prove that it is very good, first observe that $A'$ contained only one hole, hence $k_{A^*} = k_A - \max(0,d-1)-1$. Moreover, since $A^*$ is good, the degree $d_{A^*}$ of $A^*$ in $T_2^*[C']$ is the number of chords in $T'_2[C']$ crossing $e$. The number of such chords is precisely the degree $d_A$ of $A$ in $T_2$ minus the $\max(1,d)$ (corresponding to the chords crossing $e$ that were flipped to obtain $T_2^*$). Since $A$ was very good w.r.t. $T_2$, we have $k_A\geqslant d_A$, hence $k_{A^*}\geqslant d_{A^*}$ and $A^*$ is very good w.r.t. $T_2^*[C']$.
    
    By induction, we can match $k_{A^*}$ chords of the trees using at most $\frac{5}{3}k_{A^*}$ flips in total. Adding the flips applied to obtain $T_1^*$ and $T_2^*$ concludes the induction.    
\end{proof}

To complete the proof, it now remains to prove Claim~\ref{cl:very_good}.
        
\begin{proof}[Proof of Claim~\ref{cl:very_good}]
        Recall that $d\leq 3$ and that $A'$ contains only one hole $h$. We consider three cases depending on $d$ and will proceed as illustrated in Figure~\ref{fig:verygood}.
            
        \smallskip
        \noindent
        \textbf{Case $d\leq 1$:} \\
        In that case, we claim that we can flip an edge of $T_2$ crossing $e$ into $e'$.

        Since $T_1$ and $T_2$ form a nice pair and $e'$ is a chord of $T_1$, then $T_2$ has also a chord, so neither $T_1$ nor $T_2$ is a border path.
        Let $P$ be a border path incident to $h$. If $d=1$, we additionally require $P$ to contain the endpoint of a chord of $T_2$ that is in $A'$. By Observation~\ref{ob:path}, the face $f$ of $T_2$ containing $h$ has a chord $e_1$ with one endpoint in $P$. By Observation~\ref{ob:bound_hole}, the flip $e_1\rightsquigarrow h$ is valid, and then we may perform $h\rightsquigarrow e'$. Let $T'_2$ be the resulting tree. However, $T'_2$ could be obtained from $T_2$ using only the flip $e_1\rightsquigarrow e'$, since the only edge that may cross $e'$ in $T_2$ is $e_1$ (when $d=1$).

        Finally note that $e$ was not flipped, and we removed one edge of $T_2$ crossing $e$, namely $e_1$.
        
        \smallskip
        \noindent
        \textbf{Case $d = 2$:} \\        
        Let $u_1,u_2$ be the two endpoints in $A'$ of the chords $e_1,e_2$ of $T_2$ with endpoints in $A'$. Note that we may have $u_1=u_2$, but not $e_1=e_2$ since $A$ is good. Let $P_1,P_2$ be the two border paths incident with $h$. 

        If both $u_1$ and $u_2$ lie on the same border path, say $P_2$, then Lemma~\ref{lem:intborderpath}(ii) ensures that $P_2$ is internal and is incident with no other chord of $T_2$. 
        
        Otherwise, $u_1$ and $u_2$ lie on distinct border paths $P_1$ and $P_2$. Since $A$ is very good, we have $k_A\geqslant d\geqslant 2$, hence $P_1$ or $P_2$ is internal. Moreover, by Lemma~\ref{lem:intborderpath}, if both paths are internal, at least one of them (say $P_2$) is not the extra path. And if one of them (say $P_1$) is not internal, then it is the extra path, and $P_2$ is an internal path. 
        
        So, up to symmetry, we can assume that $P_2$ is an internal border path incident with only one chord of $T_2$, say $e_2$. Therefore, there exists a hole $h'\neq h$ in $A$ that is incident to $P_2$. Now by Observation~\ref{ob:path}, the face of $T_2$ containing $h'$ contains an edge with an endpoint in $P_2$, hence it is $e_2$. We then perform $e_2\rightsquigarrow h'$ in $T_2$. 

        Using Case 1, we may then perform $e_1\rightsquigarrow e'$ in the resulting tree. It then remains to flip a chord of $T_1$ on $h'$. Let $f$ be the face of $T_1$ containing $h'$. If $k_A=2$, then the only chords of $T_1$ on the boundary of $f$ are $e$ and $e'$, and we perform $e\rightsquigarrow h'$. Otherwise, since $e'$ is inclusion-wise minimal, there exists another chord $e'_1\notin\{e,e'\}$ on $f$, that we can flip on $h'$. 

        Note that in both cases, we filled a hole in both trees and created $e'$ in $T_2$. Moreover $e$ was not flipped unless $k_A=d$, and we removed two chords of $T_2$ crossing $e$ (namely $e_1$ and $e_2$).

        \smallskip
        \noindent
        \textbf{Case $d = 3$:} \\
        Let $u_1, u_2$ and $u_3$ be the three points in $A'$ of the chords $e_1,e_2,e_3$ of $T_2$ with an endpoint in $A'$. Since $A$ is very good, $e_1,e_2$ and $e_3$ are distinct.
        
        By Lemma~\ref{lem:intborderpath}, we can assume up to symmetry that $u_1$ and $u_2$ belong to $P_1$ and $u_3$ to $P_2$. (Note that $u_1$ and $u_2$ might be the same).
        Since there is at most one extra path, Lemma~\ref{lem:intborderpath} ensures that $P_2$ is internal and $e_3$ is the only chord of $T_2$ incident with $P_2$. In particular, $P_2$ is incident to a hole $h^*\neq h$ in $A$.
        
        Let us first prove that we can fill $h^*$ in $T_2$ by flipping a chord that has an endpoint inside $A$.
        By Observation~\ref{ob:path}, the face containing $h^*$ in $T_2$ has a chord on its boundary with an endpoint in $P_2$, which has to be $e_3$.
        By Observation~\ref{ob:bound_hole}, we can perform $e_3 \rightsquigarrow h^*$ in $T_2$.
    
        Let us now prove that we can fill $h^*$ in $T_1$ by flipping a chord different from $e$ and $e'$. Let $f$ be the face of $T_1$ containing $h^*$. If the boundary of $f$ has no chord besides $e$ and $e'$, then by minimality of $e'$, $A$ contains only two chords and $k_A=2$, which is impossible since $A$ is very good. Therefore there is another chord $e^*_1$ that we may flip on $h^*$ by Observation~\ref{ob:bound_hole}. 

    We may now apply Case 2 to the resulting trees $T'_1$ and $T'_2$. Observe that in total, we filled two holes in both $T_1$ and $T_2$, and created $e'$ in $T_2$ using $5$ flips. Moreover, we removed the chords $e_1,e_2,e_3$ from $T_2$ that were crossing $e$. Finally, we only flip $e$ when $k_A=2$ in $T'_1$, that is when $k_A=3$ in $T_1$, which completes the proof. 
    \end{proof}
    
 Since there is at most one extra path in a very good side, observe that Case 3 of Claim~\ref{cl:very_good} only happens once during the whole induction. Hence, with a deeper analysis, the $\frac 53k_A$ bound of Lemma~\ref{lem:very_good} can be improved to $\lceil \frac 32 k_A \rceil$. However, there is an example of a small very good side $A$ of $T_1$ w.r.t $T_2$, such that we cannot match the $k_A$ chords in $A$ of $T_1$ with chords of $T_2$ in less than $\frac 53k_A = \lceil \frac 32 k_A \rceil$ flips, see Figure~\ref{fig:tight}.  This example was the initial stone on which we built the lower bounds of Section~\ref{sec:lower}.

\begin{figure}[hbtp]
        \begin{center}
        \tikzstyle{vertex}=[circle,draw, minimum size=7pt, scale=0.5, inner sep=1pt, fill = black]
        \tikzstyle{fleche}=[->,>=latex]
        \tikzstyle{labell}=[text opacity=1, scale =1.2]
        \begin{tikzpicture}[scale=1]

            \def\firstarc{($(0,1)+({3*cos(180)},{3*sin(180)})$) arc (-180:0:3)}
            
            \nodearc(a1)(0,1)(-140:3)[vertex]();
            \nodearc(a2)(0,1)(-40:3)[vertex]();
            \nodearc(a3)(0,1)(-60:3)[vertex]();
            \nodearc(a4)(0,1)(-80:3)[vertex]();
            \nodearc(a5)(0,1)(-100:3)[vertex]();
            \nodearc(a6)(0,1)(-120:3)[vertex]();

            \begin{scope}
            \clip \firstarc;
            \fill[black!10] (a1) to[bend left = 50] (a2) to[bend left = 90] (a1);
            \draw[thick]  (a1) to[bend left = 50] (a2) to[bend left = 90] (a1); 
            \end{scope}

            \centerarc[thick, dashed](0,1)(-140:-150:3);
            \centerarc[thick, dashed](0,1)(-40:-30:3);
            
            \centerarc[thick](0,1)(-120:-100:3);
            \centerarc[thick](0,1)(-80:-60:3);
            \draw[thick] (a1) to [bend left=50] (a3);
            \draw[thick] (a3) to [bend right=40] (a6);

            \nodearc(a1)(0,1)(-140:3)[vertex]();
            \nodearc(a2)(0,1)(-40:3)[vertex]();
            \nodearc(a3)(0,1)(-60:3)[vertex]();
            \nodearc(a4)(0,1)(-80:3)[vertex]();
            \nodearc(a5)(0,1)(-100:3)[vertex]();
            \nodearc(a6)(0,1)(-120:3)[vertex]();

            \draw[thick, red] (a3) to [bend right=20] (a4);
            \draw[thick, red] (a5) to [bend right=20] (a6);
            \draw[thick, red] (a4) to  ({3*cos(-75)},0.5);
            \draw[thick, red] (a4) to  ({3*cos(-85)},0.5);
            \draw[thick, red] (a5) to  ({3*cos(-100)},0.5);
            \draw[thick, red, dashed] (a2) to  ({3*cos(-37)},-0.4);
            \draw[thick, red, dashed] (a1) to  ({3*cos(-143)},-0.4);

            \node (h3) at ({3*cos(85)},0.8) [labell, text = red]{$e_1$};
            \node (h3) at ({3*cos(100)},0.8) [labell, text = red]{$e_3$};
            \node (h3) at ({3*cos(75)},0.8) [labell, text = red]{$e_2$};
            \node (h3) at ({3*cos(110)}, -0.6) [labell]{$e^*_1$};
            \node (h3) at ({3*cos(55)},0.25) [labell]{$e = e_1'$};
            \node (h3) at ({3*cos(90)},-1.3) [labell]{$e'$};
            \nodearc(a4)(0,1)(-90:3.3)[labell]($h$);
            \nodearc(a5)(0,1)(-50:3.3)[labell]($h'$);
            \nodearc(a5)(0,1)(-130:3.3)[labell]($h^*$);

        \end{tikzpicture}
        \end{center}
        \caption{An example of a very good side which needs at least $\frac{5}{3}k_A = 5$ flips to make the two trees agree on $A$.  }
        \label{fig:tight}
\end{figure}
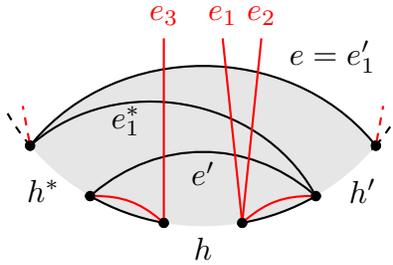

Now that we have shown how we can handle very good sides, we need to find one in a minimal counterexample $T_I, T_F$. 
However, we did not proved that a very good side always exists in $T_I, T_F$. 
We present in the next section how to transform some side which exists in $T_I, T_F$ into a very good side.

\subsection{Obtaining a very good side}\label{sec:comb}

Let $T$ and $T'$ be a nice pair of trees. The goal of this paragraph is to create a very good side for $T'$ in $T$. Recall that a good side of $T$ w.r.t. $T'$ is a side of $T$ containing no chord in $T'$.

The first step to obtain a very good side is to create a good side from a special type of side (which always exist in a minimal counterexample) we describe later on. 
A \emph{bad hole} $h$ of a side $A$ of $T$ w.r.t. $T'$ is a hole in $A$ that is also in a side $B \subsetneq A$ of $T'$, see Figure~\ref{fig:bad}. 
Note that a side of $T$ is a good side w.r.t. $T'$ if and only if it does not contain bad holes of $T$ w.r.t. $T'$. 
So in order to obtain a good side, our goal will consist in filling bad holes.

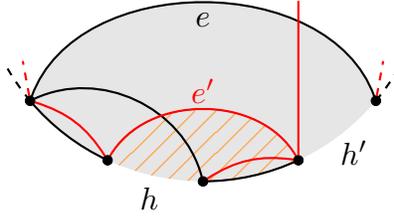
\begin{figure}[hbtp]
        \begin{center}
        \tikzstyle{vertex}=[circle,draw, minimum size=7pt, scale=0.5, inner sep=1pt, fill = black]
        \tikzstyle{fleche}=[->,>=latex]
        \tikzstyle{labell}=[text opacity=1, scale =1.2]
        \begin{tikzpicture}[scale=1]

            \def\firstarc{($(0,1)+({3*cos(180)},{3*sin(180)})$) arc (-180:0:3)}
            
            \nodearc(a1)(0,1)(-140:3)[vertex]();
            \nodearc(a2)(0,1)(-40:3)[vertex]();
            \nodearc(a3)(0,1)(-65:3)[vertex]();
            \nodearc(a4)(0,1)(-90:3)[vertex]();
            \nodearc(a5)(0,1)(-115:3)[vertex]();

            \begin{scope}
            \clip \firstarc;
            \fill[black!10] (a1) to[bend left = 70] (a2) to[bend left = 90] (a1); 
            \draw[thick]  (a1) to[bend left = 70] (a2) to[bend left = 90] (a1); 
            \draw[red, pattern=hatch, hatch size=8, hatch angle = 80, pattern color=orange!70]  (a5) to[bend left = 60] (a3) to[bend left = 90] (a5);
            \end{scope}

            \centerarc[thick, dashed](0,1)(-140:-150:3);
            \centerarc[thick, dashed](0,1)(-40:-30:3);
            
            \centerarc[thick](0,1)(-140:-115:3);
            \centerarc[thick](0,1)(-90:-65:3);
            \draw[thick] (a1) to [bend left=50] (a4);
            
            \draw[thick, red] (a1) to [bend left=20] (a5);
            \draw[thick, red] (a4) to [bend left=20] (a3);
            \draw[thick, red] (a5) to [bend left=60] (a3);

            \nodearc(a1)(0,1)(-140:3)[vertex]();
            \nodearc(a2)(0,1)(-40:3)[vertex]();
            \nodearc(a3)(0,1)(-65:3)[vertex]();
            \nodearc(a4)(0,1)(-90:3)[vertex]();
            \nodearc(a5)(0,1)(-115:3)[vertex]();

            \draw[thick, red] (a3) to  ({3*cos(-65)},0.4);
            \draw[thick, red, dashed] (a2) to  ({3*cos(-37)},-0.4);
            \draw[thick, red, dashed] (a1) to  ({3*cos(-143)},-0.4);

            \node (h3) at ({3*cos(90)},-0.8) [labell, text = red]{$e'$};
            \node (h3) at ({3*cos(90)},0.15) [labell]{$e$};
            \nodearc(a4)(0,1)(-102.5:3.3)[labell]($h$);
            \nodearc(a5)(0,1)(-52.5:3.3)[labell]($h'$);

        \end{tikzpicture}
        \end{center}
        \caption{Let $T_1$ be the black tree and $T_2$ the red tree. The hole $h$ is a bad hole of the side $A$ (in grey) w.r.t $T_2$ since it is inside the side of $e'$ included in $A$.} 
        \label{fig:bad}
\end{figure}

To generate a good side in $T$, we start from a side $A$ of $T$ that contains at least one hole that is not bad w.r.t. $T'$. Such a side always exists in a nice pair since, at least one of the two sides of the same chord of $T$ must contain a hole that is not bad w.r.t. to $T'$ (otherwise, $T'$ would contain a cycle).

Note that we will actually require more conditions on the initial side in order to conclude, namely that it is $\tau$-extremal (which will be defined later). However, this does not impact the transformation into a very good side. 

Lemma~\ref{lem:bad} explains how we fill one bad hole in $A$, and our goal is to apply it several times to obtain a good side.

\begin{lemma}\label{lem:bad}
    Let $T_1$, $T_2$ be a nice pair of trees, $A$ be a side of a chord $e$ of $T_1$ that contains at least two holes including at least one bad hole $h$ w.r.t $T_2$.

    Then, we can fill $h$ in $T_1$ by flipping a chord different from $e$ and we can fill $h$ in $T_2$ by flipping a chord with both endpoints in $A$.    
    Moreover the resulting pair of trees $T'_1$ and $T'_2$ form a nice pair and $A$ has one less bad hole.
\end{lemma}

\begin{proof}
    An illustration of the proof is given in Figure~\ref{fig:fillbad}.
    
    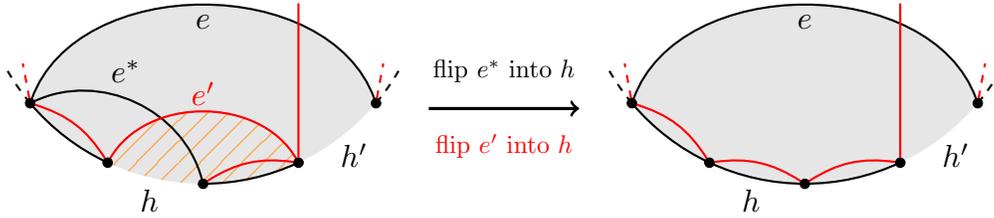
\begin{figure}[hbtp]
        \begin{center}
        \tikzstyle{vertex}=[circle,draw, minimum size=7pt, scale=0.5, inner sep=1pt, fill = black]
        \tikzstyle{fleche}=[->,>=latex]
        \tikzstyle{labell}=[text opacity=1, scale =1.2]
        \begin{tikzpicture}[scale=1]

            \def\firstarc{($(0,1)+({3*cos(180)},{3*sin(180)})$) arc (-180:0:3)}
            
            \nodearc(a1)(0,1)(-140:3)[vertex]();
            \nodearc(a2)(0,1)(-40:3)[vertex]();
            \nodearc(a3)(0,1)(-65:3)[vertex]();
            \nodearc(a4)(0,1)(-90:3)[vertex]();
            \nodearc(a5)(0,1)(-115:3)[vertex]();

            \begin{scope}
            \clip \firstarc;
            \fill[black!10] (a1) to[bend left = 70] (a2) to[bend left = 90] (a1);
            \draw[thick]  (a1) to[bend left = 70] (a2) to[bend left = 90] (a1); 
            \draw[red, thin, pattern=hatch, hatch size=8, hatch angle = 80, pattern color=orange!70]  (a5) to[bend left = 60] (a3) to[bend left = 90] (a5);
            \end{scope}

            \centerarc[thick, dashed](0,1)(-140:-150:3);
            \centerarc[thick, dashed](0,1)(-40:-30:3);
            
            \centerarc[thick](0,1)(-140:-115:3);
            \centerarc[thick](0,1)(-90:-65:3);
            \draw[thick] (a1) to [bend left=50] (a4);
            
            \draw[thick, red] (a1) to [bend left=20] (a5);
            \draw[thick, red] (a4) to [bend left=20] (a3);
            \draw[thick, red] (a5) to [bend left=60] (a3);

            \nodearc(a1)(0,1)(-140:3)[vertex]();
            \nodearc(a2)(0,1)(-40:3)[vertex]();
            \nodearc(a3)(0,1)(-65:3)[vertex]();
            \nodearc(a4)(0,1)(-90:3)[vertex]();
            \nodearc(a5)(0,1)(-115:3)[vertex]();

            \draw[thick, red] (a3) to  ({3*cos(-65)},0.4);
            \draw[thick, red, dashed] (a2) to  ({3*cos(-37)},-0.4);
            \draw[thick, red, dashed] (a1) to  ({3*cos(-143)},-0.4);

            \node (h3) at ({3*cos(90)},-0.8) [labell, text = red]{$e'$};
            \node (h3) at ({3*cos(110)}, -0.5) [labell]{$e^*$};
            \node (h3) at ({3*cos(90)},0.15) [labell]{$e$};
            \nodearc(a4)(0,1)(-102.5:3.3)[labell]($h$);
            \nodearc(a5)(0,1)(-52.5:3.3)[labell]($h'$);

        \tikzset{xshift=3cm}

            \node (h0) at (1, -0.5) [labell, scale=0.8]{flip $e^*$ into $h$};
            \draw[->, very thick] (0,-1) to (2,-1);
            \node (h0) at (1, -1.5) [labell, scale=0.8, text = red]{flip $e'$ into $h$};
            
        \tikzset{xshift=5cm}

            \def\firstarc{($(0,1)+({3*cos(180)},{3*sin(180)})$) arc (-180:0:3)}
            
            \nodearc(a1)(0,1)(-140:3)[vertex]();
            \nodearc(a2)(0,1)(-40:3)[vertex]();
            \nodearc(a3)(0,1)(-65:3)[vertex]();
            \nodearc(a4)(0,1)(-90:3)[vertex]();
            \nodearc(a5)(0,1)(-115:3)[vertex]();

            \begin{scope}
            \clip \firstarc;
            \fill[black!10] (a1) to[bend left = 70] (a2) to[bend left = 90] (a1);
            \draw[thick]  (a1) to[bend left = 70] (a2) to[bend left = 90] (a1); 
            \end{scope}

            \centerarc[thick, dashed](0,1)(-140:-150:3);
            \centerarc[thick, dashed](0,1)(-40:-30:3);
            
            \centerarc[thick](0,1)(-140:-115:3);
            \centerarc[thick](0,1)(-90:-65:3);
            \centerarc[thick](0,1)(-90:-115:3);
            
            \draw[thick, red] (a1) to [bend left=20] (a5);
            \draw[thick, red] (a4) to [bend left=20] (a3);
            \draw[thick, red] (a5) to [bend left=20] (a4);

            \nodearc(a1)(0,1)(-140:3)[vertex]();
            \nodearc(a2)(0,1)(-40:3)[vertex]();
            \nodearc(a3)(0,1)(-65:3)[vertex]();
            \nodearc(a4)(0,1)(-90:3)[vertex]();
            \nodearc(a5)(0,1)(-115:3)[vertex]();

            \draw[thick, red] (a3) to  ({3*cos(-65)},0.4);
            \draw[thick, red, dashed] (a2) to  ({3*cos(-37)},-0.4);
            \draw[thick, red, dashed] (a1) to  ({3*cos(-143)},-0.4);

            \node (h3) at ({3*cos(90)},0.15) [labell]{$e$};
            \nodearc(a4)(0,1)(-102.5:3.3)[labell]($h$);
            \nodearc(a5)(0,1)(-52.5:3.3)[labell]($h'$);

        \end{tikzpicture}
        \end{center}
        \caption{The side $A$ (in grey) of the chord $e$ in the tree $T_1$ (in black) contains two holes $h$ and $h'$, with $h$ a bad hole of $A$ w.r.t the tree $T_2$ (in red). We can fill $h$ in $T_1$ by flipping $e^*$ and we can fill $h$ in $T_2$ by flipping $e'$. }
        \label{fig:fillbad}
\end{figure}
    
    Let us first prove that we can fill $h$ in $T_1$ by flipping a chord different from $e$.
    Let $f$ be the face containing $h$ in $T_1$. 
    Note that the boundary of $f$ is included in $A$. 
    If $e$ is the only chord on the boundary of $f$, then $f$ is the only face whose boundary is included in $A$ and $h$ is the only hole of $A$, a contradiction. 
    Otherwise, let $e^*$ be a chord of $T_1$ different from $e$ that is on the boundary of $f$.
    By Observation~\ref{ob:bound_hole}, we can perform $e^* \rightsquigarrow h$ in $T_1$.

    Let us now prove that we can fill $h$ in $T_2$ by flipping a chord in $A$.
    Let $f'$ be the face containing $h$ in $T_2$. Since $h$ is a bad hole of $A$ w.r.t $T_2$, the face $f'$ is included in $A$.
    Let $e'$ be a chord of $T_2$ on the boundary of $f'$, and note that $e'$ is a chord in $A$.
    By Observation~\ref{ob:bound_hole}, we can perform $e' \rightsquigarrow h$ in $T_2$.
    Note that $T_1',T_2'$ is still nice since we created a common border edge and no common chord.

\end{proof}

Recall that the side $A$ of $T$ contains a good hole. Therefore, each time we update $T,T'$ applying Lemma~\ref{lem:bad} on $A$, the good holes in $A$ are not filled, which ensures we can repeatedly apply the lemma until no bad hole remains. After this process, we filled $m<k_A$ bad holes w.r.t. $T'$ in $2m$ flips, and $A$ is now a good side of $T$ w.r.t. $T'$. Let us now explain how we can transform $A$ into a very good side.   

Recall that $A$ being not very good simply means that there are too many chords of $T'$ crossing the unique chord $e$ on the boundary of $A$. The goal of Lemma~\ref{lem:good} is to remove these extra crossings. To obtain a very good side, we will apply it iteratively until we reach the right amount of chords.

\begin{lemma}\label{lem:good}
    Let $T_1$, $T_2$ be a nice pair of trees, $A$ be a good side of a chord $e$ of $T_1$ w.r.t $T_2$ which is not very good w.r.t. $T_2$.
    Then there exists a hole $h$ not in $A$ such that: (i) we can fill $h$ in $T_1$ by flipping a chord distinct from $e$ and (ii) we can fill $h$ in $T_2$ by flipping a chord crossing $e$.

    Moreover the resulting pair of trees after these two flips is still nice.
\end{lemma}

\begin{proof} 
The proof is illustrated in Figure~\ref{fig:removextra}.

\begin{figure}[hbtp]
        \begin{center}
        \tikzstyle{vertex}=[circle,draw, minimum size=7pt, scale=0.5, inner sep=1pt, fill = black]
        \tikzstyle{fleche}=[->,>=latex]
        \tikzstyle{labell}=[text opacity=1, scale =1.2]
        \begin{tikzpicture}[scale=1]

            \nodeellipse(a1)(0,0)(-20:2.5:1.5)[vertex]();
            \nodeellipse(a2)(0,0)(20:2.5:1.5)[vertex]();
            \nodeellipse(a3)(0,0)(-60:2.5:1.5)[vertex]();
            \nodeellipse(a4)(0,0)(-100:2.5:1.5)[vertex]();
            \nodeellipse(a5)(0,0)(-140:2.5:1.5)[vertex]();
            \nodeellipse(a6)(0,0)(180:2.5:1.5)[vertex]();
            \nodeellipse(a7)(0,0)(60:2.5:1.5)[vertex]();
            \nodeellipse(a8)(0,0)(100:2.5:1.5)[labell]();
            \nodeellipse(a9)(0,0)(140:2.5:1.5)[vertex]();

            \begin{scope}
            \clip (0,0) ellipse (2.5cm and 1.5cm);
            \fill[black!10]  (a3) to [bend left = 10] (a7) to[bend left = 90] (a2) to[bend left=90] (a3); 
            \fill[thin, pattern=hatch, hatch size=8, hatch angle = 20, pattern color=black!50]   (a6) to [bend right = 90] (a5) to[bend left = 20] (a4) to[bend right =10] (a9) to[bend left=40] (a6); 
            \end{scope}
    
            \nodeellipse(a1)(0,0)(-20:2.5:1.5)[vertex]();
            \nodeellipse(a2)(0,0)(20:2.5:1.5)[vertex]();
            \nodeellipse(a3)(0,0)(-60:2.5:1.5)[vertex]();
            \nodeellipse(a4)(0,0)(-100:2.5:1.5)[vertex]();
            \nodeellipse(a5)(0,0)(-140:2.5:1.5)[vertex]();
            \nodeellipse(a6)(0,0)(180:2.5:1.5)[vertex]();
            \nodeellipse(a7)(0,0)(60:2.5:1.5)[vertex]();
            \nodeellipse(a8)(0,0)(100:2.5:1.5)[labell]();
            \nodeellipse(a9)(0,0)(140:2.5:1.5)[vertex]();

            \draw[thick] (a3) to [bend left = 10] (a7);
            \draw[thick] (a3) to [bend left = 10] (a2);
            \draw[thick] (a4) to [bend right = 10] (a9);

            \draw[thick,dashed] (a9) to[bend left = 40] (a6);
            \draw[thick,dashed] (a5) to[bend left = 20] (a4);
            
            \draw[thick] (a1) arc [x radius = 2.5cm, y radius = 1.5cm, start angle = -20, end angle = 20] ;
            
            \draw[thick, red] (a1) to (a5);
            \draw[thick, red] (a1) to (a9);
            \draw[thick, red] (a2) to (a8);
            \draw[thick, red] (a1) to[bend left = 20] (a2);
            \draw[thick,dashed, red] (a9) to[bend left=5] (a6);

            \node (h3) at (30:1) [labell]{$e$};
            \node (h3) at (0,0) [labell, text=red]{$e_i$};
            \nodeellipse(h9)(0,0)(-160:2.9:1.8)[labell]($h$);
            \nodeellipse(h9)(0,0)(165:1:1)[labell]($e^*$);

      \tikzset{xshift=3.5cm}

            \node (h0) at (1, 0.5) [labell, scale=0.8]{flip $e^*$ into $h$};
            \draw[->, very thick] (0,0) to (2,0);
            \node (h0) at (1, -0.5) [labell, scale=0.8, text = red]{flip $e_i$ into $h$};
            
        \tikzset{xshift=5.5cm}

            \nodeellipse(a1)(0,0)(-20:2.5:1.5)[vertex]();
            \nodeellipse(a2)(0,0)(20:2.5:1.5)[vertex]();
            \nodeellipse(a3)(0,0)(-60:2.5:1.5)[vertex]();
            \nodeellipse(a4)(0,0)(-100:2.5:1.5)[vertex]();
            \nodeellipse(a5)(0,0)(-140:2.5:1.5)[vertex]();
            \nodeellipse(a6)(0,0)(180:2.5:1.5)[vertex]();
            \nodeellipse(a7)(0,0)(60:2.5:1.5)[vertex]();
            \nodeellipse(a8)(0,0)(100:2.5:1.5)[labell]();
            \nodeellipse(a9)(0,0)(140:2.5:1.5)[vertex]();

            \begin{scope}
            \clip (0,0) ellipse (2.5cm and 1.5cm);
            \fill[black!10]  (a3) to [bend left = 10] (a7) to[bend left = 90] (a2) to[bend left=90] (a3); 
            \end{scope}

            \nodeellipse(a1)(0,0)(-20:2.5:1.5)[vertex]();
            \nodeellipse(a2)(0,0)(20:2.5:1.5)[vertex]();
            \nodeellipse(a3)(0,0)(-60:2.5:1.5)[vertex]();
            \nodeellipse(a4)(0,0)(-100:2.5:1.5)[vertex]();
            \nodeellipse(a5)(0,0)(-140:2.5:1.5)[vertex]();
            \nodeellipse(a6)(0,0)(180:2.5:1.5)[vertex]();
            \nodeellipse(a7)(0,0)(60:2.5:1.5)[vertex]();
            \nodeellipse(a8)(0,0)(100:2.5:1.5)[labell]();
            \nodeellipse(a9)(0,0)(140:2.5:1.5)[vertex]();

            \draw[thick] (a3) to [bend left = 10] (a7);
            \draw[thick] (a3) to [bend left = 10] (a2);

            \draw[thick,dotted] (a9) to[bend left] (a6);
            \draw[thick,dotted] (a5) to[bend left] (a4);
            
            \draw[thick] (a5) arc [x radius = 2.5cm, y radius = 1.5cm, start angle = -140, end angle =-180] ;
            \draw[thick] (a1) arc [x radius = 2.5cm, y radius = 1.5cm, start angle = -20, end angle = 20] ;
            
            \draw[thick, red] (a1) to (a5);
            \draw[thick, red] (a2) to (a8);
            \draw[thick, red] (a1) to[bend left = 20] (a2);
            \draw[thick, red] (a5) to[bend right = 20] (a6);
            \draw[thick,dashed, red] (a9) to[bend left=5] (a6);

            \node (h3) at (30:1) [labell]{$e$};
            \nodeellipse(h9)(0,0)(-160:2.9:1.8)[labell]($h$);

        \end{tikzpicture}
        \end{center}
        \caption{An example of a nice pair $T_1$ (in black) and $T_2$ (in red). We can fill $h$ in $T_1$ by flipping $e^*$, a chord on the face (in grey lines) of $T_1$ associated to $h$, and we can fill $h$ in $T_2$ by flipping $e_i$, a chord on the face of $T_2$ associated to $h$ which crosses $e$.  }
        \label{fig:removextra}
\end{figure}
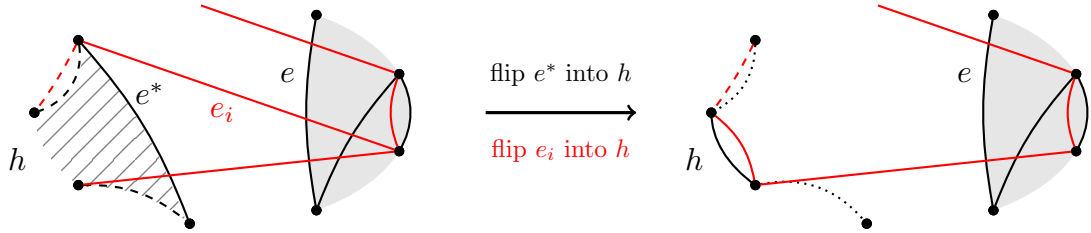    

    Let $\ell > k_A$ be the degree of $A$ in $T_2$. And let us denote by $B$ be the other side of $e$.
    Since $A$ is good w.r.t. $T_2$, $A$ contains no chord of $T_2$, and hence every chord $e_0$ of $T_2$ with an endpoint inside $A$ for $e_0$ crosses $e$, so has its other endpoint in $C \setminus A$. 
    Thus, there are $\ell$ chords $e_1,\dots, e_\ell$ of $T_2$ that cross $e$. 
    These $\ell$ chords split the convex hull of $C$ into $\ell+1$ parts, hence there are at least $\ell+1$ faces of $T_2$ that contain at least one chord in $e_1, \dots, e_\ell$ on their boundary. 
    Now consider the holes contained in these $\ell + 1$ faces of $T_2$. 
    Since $\ell > k_A$, there are at least two such holes, say $h$ and $h'$, that are not in $A$.
  
    Let $f$ and $f'$ be the faces of $T_1$ containing $h$ and $h'$. 
    Since the boundaries of $f$ and $f'$ are both included in $B$, $e$ cannot lie in both $f$ and $f'$. So, up to symmetry, we may assume $e$ is not on the boundary of $f$. 
    By Observation~\ref{ob:bound_hole}, we can perform $e^* \rightsquigarrow h$ in $T_1$ with $e^*$ a chord of $T_1$ that is on the boundary of $f$.
    By definition of $h$, there is a chord $e_i$ on the face of $T_2$ containing $h$ and we can perform $e_i \rightsquigarrow h$ in $T_2$ by Observation~\ref{ob:bound_hole}.
    Since we only flipped chords to border edges, the resulting pair of trees is still nice.
\end{proof}

Let $d$ be the degree of $A$ in $T'$. Applying $d-k_A$ times Lemma~\ref{lem:good} transforms $A$ into a very good side w.r.t. $T'$ using $2(d - k_A)$ flips. 

To summarize, using Lemma~\ref{lem:bad} and Lemma~\ref{lem:good}, we are able to transform any side $A$ containing a good hold into a very good side. However, at each step, we use $2$ flips and fill one hole. Once $A$ becomes very good side, our goal is to apply Lemma~\ref{lem:very_good} to match $k_A$ edges in $T$ and $T'$ in $\frac{5k_A}{3}$ flips.

\subsection{Bounding the number of flips}

In order to conclude the proof of Theorem~\ref{thm:mainthm}, we need to make sure that we save enough using Lemma~\ref{lem:very_good} to compensate for the expensive use of Lemmas~\ref{lem:bad} and~\ref{lem:good}. In other words, we want to ensure that we will not use too many flips to obtain a very good side relative to the number of holes in the resulting side. This is why we need to start from a side whose number of bad holes and degree are not too large compared to its number of holes. 

Let $T,T'$ be a nice pair of trees and $\tau > 2$. We say a side $A$ of a chord $e$ of $T$ is \emph{$\tau$-extremal} for a tree $T'$ if the degree of $A$ in $T'$ is at most $\tau \cdot  k_A$, and, for every side $A' \subsetneq A$ of $T'$, the degree of $A'$ in $T$ is more than $\tau \cdot k_{A'}$. First, we prove that such a side exists in a minimal counterexample.

\begin{lemma}\label{cl:existtau}
    Let $T_1$ and $T_2$ be a nice pair of trees that are not border paths, then either $T_1$ or $T_2$ contains a $\tau$-extremal side.
\end{lemma}
\begin{proof}
    Let $e$ be a chord of $T_1$, and $A$ and $B$ be the two sides of $e$. 
    Recall that $k_A + k_B$ is the number of holes in $T_1$ (thus in $T_2$).
    Since the number of holes in $T_2$ is equal to the number of chords in $T_2$ minus $1$, $T_2$ contains $k_A + k_B - 1$ chords.
    In particular, there are $2(k_A + k_B - 1)$ endpoints of chords of $T_2$, hence by symmetry we get that $A$ has degree at most $2k_A \leq \tau  \cdot k_A$ in $T_2$.
    
    If $A$ is not $\tau$-extremal, there is a side $A' \subsetneq A$ of $T_2$ of degree at most $\tau \cdot k_{A'}$ in $T_1$. We now replace $A$ by $A'$, swap $T_1$ and $T_2$ and iterate this process until we find a $\tau$-extremal side. Note that this terminates since each time $A'\subsetneq A$.
\end{proof}

We now show that a $\tau$-extremal side has not too many bad holes, which we use later for bounding the number of flips in our process. 

\begin{lemma}\label{lem:boundbadholes}
Let $T_1$ and $T_2$ be a nice pair of trees, and $A$ be a $\tau$-extremal side of a chord $e$ of $T$. Then the side $A$ contains at most $\frac{2}{\tau} k_A$ bad holes w.r.t. $T_2$.
\end{lemma}
\begin{proof}
Let $m$ be the number of bad holes in $A$ w.r.t. $T_2$. 
For each bad hole $h$ w.r.t. $T_2$ in $A$, let $e_h$ be the chord of $T'$ such that the side of $e_h$ included in $A$ is inclusion-wise maximal (with possibly $e_h=e_{h'}$ for distinct bad holes $h,h'$). 
Then, we define $\mathcal{A}$ as the set formed by the sides $A'\subseteq A$ of all the chords $e_h$.

Note that these sides do not overlap, so no hole belongs to two different sides in $\A$.
Thus, each bad hole w.r.t. $T_2$ is contained in exactly one side $A' \in \mathcal{A}$. 
Further, each $A' \in \A$ contains exactly $k_{A'}$ holes of $T_2$, that are all bad w.r.t $T_2$. 
Thus, we get that $m=\sum_{A' \in \mathcal{A}} k_{A'}$. 
Since $A$ is $\tau$-extremal, the degree of $A' \in \A$ in $T_1$ is more than $\tau k_{A'}$.
Since there are $k_A$ chords of $T_1$ in $A$, the sum of the degrees of the sides $A'\in \A$ in $T_1$ is at most $2k_A$. But then:
    \begin{align*}
    2k_A \geq \sum_{A' \in \mathcal{A}} \tau \cdot k_{A'} = \tau m.
    \end{align*}
    Rearranging the equation proves the lemma.
\end{proof}

We are now ready to prove Theorem~\ref{thm:mainthm}. Let us first explain the intuition of the proof.
By Lemma~\ref{cl:existtau}, a $\tau$-extremal side exists in a minimal counterexample. Moreover, informally speaking, a $\tau$-extremal side does not have too large degree (by definition) and does not contain too many bad holes by Lemma~\ref{lem:boundbadholes}.
Thus,  we can obtain a very good side from a $\tau$-extremal side using not too many flips.
And, we will use this idea along all we have proved on very good sides and minimal counterexamples to prove the upper bound on a minimal flip sequence. 

\mainthm*

\begin{proof}

Assume by contradiction Theorem~\ref{thm:mainthm} does not hold and let us consider  a minimum counterexample $T_I$ and $T_F$.
By Corollary~\ref{cor:nice}, $T_I$ and $T_F$ form a nice pair of trees.
Consider a $\tau$-extremal side $A$ in the counterexample. 
By symmetry, we can assume that $A$ is a side of some chord $e$ of $T_I$.
By Lemma~\ref{lem:boundbadholes}, the side $A$ contains $ m \leq \frac{2}{\tau} k_A$ bad holes w.r.t $T_F$. In particular $A$ contains a good hole, and we can apply $m$ times Lemma~\ref{lem:bad} to fill every bad hole of $A$ w.r.t. $T_F$ using $2m$ flips.

By Lemma~\ref{lem:bad}, the resulting pair of trees $T_I'$ and $T_F'$ form a nice pair of trees and the side $A'$ of the chord $e$ of $T_I'$ is a good side of $T_I'$ w.r.t $T_F'$ of size $k_{A'} = k_A - m$ such that the degree of $A'$ in $T_F'$ is at most $\tau k_A - 2m$.

Now, we apply $(\tau - 1)k_A - m$ times Lemma~\ref{lem:good}, until the degree of $A'$ in $T_F''$ becomes at most $k_{A'}$.
Again, note that each time by Lemma~\ref{lem:good}, we are left with a nice pair of trees where $e$ did not change.
After $2((\tau - 1)k_A - m)$ flips, the resulting pair of trees $T_I''$ and $T_F''$ form a nice pair of trees and the side $A'' = A$ of the chord $e$ of $T_I''$ is a very good side of $T_I''$ w.r.t. $T_F''$ of size $k_{A''} = k_{A'}= k_A - m$.

By Lemma~\ref{lem:very_good}, we can match $k_{A''}$ chords of the trees using at most $\frac{5}{3} k_{A''}$ flips. Let $T_I^*$ and $T_F^*$ be the resulting pair of trees. We have that:
    \begin{align*}
    \delta(T_I^*, T_ F^*) &= \delta(T_I'', T_F'')-k_{A''}= \delta(T_I'', T_F'') - (k_{A}-m) \\
                          &= \delta(T_I', T_F') - ((\tau - 1)k_A - m) - (k_A - m) \\
                          &= \delta(T_I, T_F) - m - ((\tau - 1)k_A - m) - (k_A - m) \\
                          &= \delta(T_I, T_F) - \tau k_A + m
    \end{align*}

Since  $\tau k_A - m > 0$, $\delta(T_I^*, T_F^*) < \delta(T_I, T_F)$ and by minimality, there exists a flip sequence between $T_I^*$ and $T_F^*$ of length at most $c\delta(T_I^*, T_F^*)$.  
Thus, we have a flip sequence between $T_I$ and $T_F$ of length at most:
    \begin{align*}
     &  c\delta(T_I^*, T_F^*) + 2m + 2((\tau - 1)k_A - m)  + \frac{5}{3} (k_A - m)  \\
     &=  c\delta(T_I, T_F) + (2 - c)\tau k_A - 2k_A + cm + \frac{5}{3} \left(k_A - m\right) \\
 &=  c\delta(T_I, T_F) + (2 - c)\tau k_A - \frac{1}{3}k_A + \left(c - \frac{5}{3}\right)m \\
  &\leq  c\delta(T_I, T_F) + \left((2 - c)\tau - \frac{1}{3} + \left( c - \frac{5}{3} \right) \frac{2}{\tau} \right) k_A     
    \end{align*}

For $\tau = 2 + \sqrt{2}$ and $c = \frac{1}{12}(22 + \sqrt{2})$, we get a flip sequence between $T_I$ and $T_F$ of length at most $c\delta(T_I, T_F)$, a contradiction.
\end{proof}


\section{Lower bounds}\label{sec:lower}

The goal of this section is to prove the different lower bounds. 

For each model, we will give a family of pairs of trees that satisfy the corresponding theorem. The proofs of the three theorems share a similar structure and rely on counting arguments. To give the flavor, we start with the simplest construction, and prove Theorem~\ref{thm:lb_ncflips} in Section~\ref{sec:lowerncflip}. We then proceed with the more involved case of flips by proving Theorem~\ref{thm:lb_flips} in Section~\ref{sec:lowerflip}. Finally, we prove Theorem~\ref{thm:lb_rotations} in Section~\ref{sec:lowerrotation} using a different counting method.

\subsection{Non-crossing flips}\label{sec:lowerncflip}

As a warm-up before proving the other items which are harder, let us prove prove Theorem~\ref{thm:lb_ncflips}.

\lbncflips*

In particular, the proof of Theorem~\ref{thm:lb_flips} will follow the same scheme but the construction and proofs will be more technical.

\begin{figure}[hbtp]
        \begin{center}
        \tikzstyle{vertex}=[circle,draw, minimum size=7pt, scale=0.5, inner sep=1pt, fill = black]
        \tikzstyle{fleche}=[->,>=latex] 
        \tikzstyle{labell}=[text opacity=1, scale =0.9]
        \begin{tikzpicture}[scale=0.9]

            \node (a2) at (-140:1.4) [vertex] {};
            \node (a4) at (140:1.4) [vertex] {};
            \node (a5) at (-40:1.4) [vertex] {};
            \node (a7) at (40:1.4) [vertex] {};

            \node (b2) at (-140:1.8) [labell] {$v_2$};
            \node (b4) at (140:1.8) [labell] {$v_1$};
            \node (b5) at (-40:1.8) [labell] {$v_4$};
            \node (b7) at (40:1.8) [labell] {$v_3$};
            
            \draw (a2) to (a4);
            \draw (a5) to (a7);
            \draw (a5) to (a4);

            \draw (a2) to [bend right = 20](a4) [red];
            \draw (a7) to [bend right = 20](a5) [red];

            \draw (a2) to (a7) [red];

        \tikzset{xshift=7.5cm}
    
        \node (i3) at ({3*cos(3*(360/14) + (360/28))}, 1)[vertex]{};
        \node (i6) at ({3*cos(6*(360/14) + (360/28))}, 1)[vertex]{};
        \node (i7) at ({3*cos(7*(360/14) + (360/28))}, -1)[vertex]{};
        \node (i13) at ({3*cos(13*(360/14) + (360/28))}, -1)[vertex]{};
        \node (i14) at ({3*cos(14*(360/14) + (360/28))}, 1)[vertex]{};
        \node (i10) at ({3*cos(10*(360/14) + (360/28))}, -1)[vertex]{};

        \node (h3) at ({3*cos(3*(360/14) + (360/28))}, 1.4)[labell]{$v_3^i = v_1^{i+1}$};
        \node (h10) at ({3*cos(10*(360/14) + (360/28))}, -1.4)[labell]{$v_4^i = v_2^{i+1}$};
        \node (h6) at ({3*cos(6*(360/14) + (360/28)) - 0.2}, 1.4)[labell]{$v_1^i$};
        \node (h7) at ({3*cos(7*(360/14) + (360/28)) - 0.2}, -1.4)[labell]{$v_2^i$};
        \node (h13) at ({3*cos(13*(360/14) + (360/28)) + 0.2}, -1.4)[labell]{$v_3^{i+1}$};
        \node (h14) at ({3*cos(14*(360/14) + (360/28)) + 0.2}, 1.4)[labell]{$v_4^{i+1}$};

        \draw[red] (i3) to[bend right=10] (i10);
        \draw (i13) to (i14) ;
        \draw (i6) to (i10) ;
        \draw (i3) to (i13) ;
        \draw (i3) to[bend left=10] (i10);
        \draw (i7) to (i6) ;

        \draw (i6) to (i10) ;
        \draw[red] (i7) to (i3);
        \draw[red] (i10) to (i14);
        \draw[red] (i6) to[bend left = 20] (i7) ;
        \draw[red] (i13) to[bend left = 20] (i14) ;

        \draw[red, dashed] (i7) to ({3*cos(7*(360/14) + (360/28)) - 0.3}, -0.4);
        \draw[dashed] (i6) to ({3*cos(6*(360/14) + (360/28)) - 0.3}, 0.4);
        \draw[red, dashed] (i13) to ({3*cos(13*(360/14) + (360/28)) + 0.3}, -0.4);
        \draw[dashed] (i14) to ({3*cos(14*(360/14) + (360/28)) + 0.3}, 0.4);
            
        \begin{scope}
            \clip (0,0) ellipse (3cm and 2cm);
        \end{scope}

        \end{tikzpicture}
        \end{center}
        \caption{On the left, the tree $T_1$ in black and the tree $T_1'$ in red. On the right, the subgraph induced by $C_i$ and $C_{i+1}$ in a pair $T_k$ (in black) and $T_k'$ (in red).  }
        \label{fig:lowernc}
    \end{figure}

Let us denote by $T_1$ and $T'_1$ the pair of non-crossing spanning trees on a convex set $C$ of size $4$ represented in Figure~\ref{fig:lowernc}. Note that we have $\delta(T_1,T_1')=1$.

For every $k$, we denote by $T_k, T_k'$ the pair of non-crossing spanning trees obtained by taking $k$ disjoint copies of $T_1,T_1'$ and identifying the points $v_3$ and $v_4$ of the $i$-th copy  respectively with the points $v_1$ and $v_2$ in the $(i+1)$-th copy for $i<k$.
We define $C_i$ as the set of points of the $i$-th copy, and $v_j^i$ the point corresponding to $v_j$ in $C_i$. In particular, the sets of points $C_i$'s are not disjoint since $C_i$ and $C_{i+1}$ intersect on $v_1^i=v_3^{i+1}$ and $v_2^i=v_4^{i+1}$. Finally observe that $\delta(T_k, T_k') = k$ for all $k \geq 1$.

\medskip

Recall that we can always transform a tree $T$ into another tree $T'$ using at most $2\delta(T,T')$ non-crossing flips by flipping edges of the symmetric difference into border edges (with an iterative application of Lemma~\ref{lem:flipborder}). The rest of the proof of Theorem~\ref{thm:lb_ncflips} consists in proving by induction on $k$ that this strategy yields a minimal non-crossing flip sequence. First, we can easily see that one cannot transform $T_1$ into $T'_1$ with one non-crossing flip, which proves the case $k=1$.
 
For the induction, consider an integer $k>1$ and assume that for $\ell < k$, a minimal non-crossing flip sequence from $T_\ell$ to $T_\ell'$ has length at least $2\ell$.
Let us first remark that the following holds.
    
\begin{lemma}\label{lem:LB_ncflips_chordschange}
    If a non-crossing flip sequence $\S$ from $T_k$ to $T_k'$ does not modify at least one common chord $e$, then $\S$ has length at least $2k$.
\end{lemma}
\begin{proof}
    By construction, there exists $i < k$ such that $e$ is a chord with both endpoints in $C_i$ and $C_{i+1}$. 
    Let $A$ and $B$ be the two sides of $e$ such that $C_i \subseteq A$ and $C_{i+1} \subseteq B$.
    Since $e$ is not modified during $\S$, no step in $\S$ removes an edge in $A$ to add an edge in $B$, or conversely (otherwise one side would not be connected anymore). 

     We can partition the sequence $\S$ into two sequences $\S_A$ and $\S_B$ where $\S_A$ is restricted to flips between points in $A$ and $\S_B$ flips between points in $B$.
     Since $e$ belongs to the tree at any step of $\S$, we can first perform all the non-crossing flips in $\S_A$ and then all the non-crossing flips in $\S_B$ while keeping connectivity.
    Note that both $T_k[A],T_k'[A]$ and $T_k[B],T_k'[B]$ induce a copy of $T_i,T_i'$ and $T_{k-i},T_{k-i}'$. 
    By induction, $|\S_A| \geq 2i$ and $|\S_B| \geq 2(k-i)$. Thus, $\S$ has length at least $2k$.
\end{proof}

\begin{lemma}\label{lem:LB_ncflips_nochordschange}
    If a non-crossing flip sequence $\S$ from $T_k$ to $T_k'$ modifies every common chord, then $\S$ has length at least $2k$.
\end{lemma}
\begin{proof}
    In the non-crossing flip sequence $\S$, we have to remove and add every common chord of $T_k$ and $T'_k$ (there are $k-1$ of them), to remove chords of $T_k \setminus T_k'$ and create chords of $T_k' \setminus T_k$. Thus, $\S$ add or remove at least $2k-2$ edges (and then contains at least $2k - 1$ non-crossing flips). 
    Since, at the first step, the first non-crossing flip cannot directly create a chord of $T'_k$, there is also an edge $e' \notin T_k\cup T'_k$ that is not described above that has to appear and be removed in $\S$. 
    Therefore $\S$ contains at least $2k$ non-crossing flips.
\end{proof}

\subsection{Flips}\label{sec:lowerflip}

The goal of this part is to prove Theorem~\ref{thm:lb_flips}.

\lbflips*

The proof also holds by induction but (i) the construction has to be different and (ii) in order to prove that the result holds, one has to analyze it with more involved arguments.

\begin{figure}[hbtp]
        \begin{center}
        \tikzstyle{vertex}=[circle,draw, minimum size=7pt, scale=0.5, inner sep=1pt, fill = black]
        \tikzstyle{fleche}=[->,>=latex] 
        \tikzstyle{labell}=[text opacity=1, scale =0.9]
        \begin{tikzpicture}[scale=0.9]

            \node (a1) at (-110:2) [vertex] {};
            \node (a2) at (-160:2) [vertex] {};
            \node (a3) at (-70:2) [vertex] {};
            \node (a4) at (160:2) [vertex] {};
            \node (a5) at (-20:2) [vertex] {};
            \node (a6) at (110:2) [vertex] {};
            \node (a7) at (20:2) [vertex] {};
            \node (a8) at (70:2) [vertex] {};

            \node (b1) at (-110:2.4) [labell] {$v_5$};
            \node (b2) at (-160:2.4) [labell] {$v_2$};
            \node (b3) at (-70:2.4) [labell] {$v_6$};
            \node (b4) at (160:2.4) [labell] {$v_1$};
            \node (b5) at (-20:2.4) [labell] {$v_8$};
            \node (b6) at (110:2.4) [labell] {$v_3$};
            \node (b7) at (20:2.4) [labell] {$v_7$};
            \node (b8) at (70:2.4) [labell] {$v_4$};
            
            \draw (a1) to (a3);
            \draw (a2) to (a4);
            \draw (a6) to (a8);
            \draw (a5) to (a7);
            \draw (a5) to (a6);
            \draw (a5) to (a4);
            \draw (a3) to (a4);

            \draw (a1) to [bend left = 20](a3) [red];
            \draw (a2) to [bend right = 20](a4) [red];
            \draw (a7) to [bend right = 20](a5) [red];
            \draw (a8) to [bend left = 20](a6) [red];

            \draw (a1) to (a7) [red];
            \draw (a2) to (a8) [red];
            \draw (a8) to (a1) [red];

        \tikzset{xshift=7.5cm}

        \foreach \s in {1,...,5}{(\node (i\s) at ({3*cos(\s*(360/14) + (360/28))}, 2)[vertex]{};};
        \node (i6) at ({3*cos(6*(360/14) + (360/28))}, 1)[vertex]{};
        \node (i7) at ({3*cos(7*(360/14) + (360/28))}, -1)[vertex]{};
        \node (i13) at ({3*cos(13*(360/14) + (360/28))}, -1)[vertex]{};
        \node (i14) at ({3*cos(14*(360/14) + (360/28))}, 1)[vertex]{};
        \foreach \s in {8,...,12}{(\node (i\s) at ({3*cos(\s*(360/14) + (360/28))}, -2)[vertex]{};};

        \node (h1) at ({3*cos(1*(360/14) + (360/28))}, 2.4)[labell]{$v_6^{i+1}$};
        \node (h2) at ({3*cos(2*(360/14) + (360/28))}, 2.4)[labell]{$v_5^{i+1}$};
        \node (h3) at ({3*cos(3*(360/14) + (360/28))}, 2.4)[labell]{$v_7^i = v_2^{i+1}$};
        \node (h4) at ({3*cos(4*(360/14) + (360/28))}, 2.4)[labell]{$v_4^i$};
        \node (h5) at ({3*cos(5*(360/14) + (360/28))}, 2.4)[labell]{$v_3^i$};
        \node (h8) at ({3*cos(8*(360/14) + (360/28))}, -2.4)[labell]{$v_5^i$};
        \node (h9) at ({3*cos(9*(360/14) + (360/28))}, -2.4)[labell]{$v_6^i$};
        \node (h10) at ({3*cos(10*(360/14) + (360/28))}, -2.4)[labell]{$v_8^i = v_1^{i+1}$};
        \node (h11) at ({3*cos(11*(360/14) + (360/28))}, -2.4)[labell]{$v_3^{i+1}$};
        \node (h12) at ({3*cos(12*(360/14) + (360/28))}, -2.4)[labell]{$v_4^{i+1}$};
        \node (h6) at ({3*cos(6*(360/14) + (360/28)) - 0.2}, 1.4)[labell]{$v_1^i$};
        \node (h7) at ({3*cos(7*(360/14) + (360/28)) - 0.2}, -1.4)[labell]{$v_2^i$};
        \node (h13) at ({3*cos(13*(360/14) + (360/28)) + 0.2}, -1.4)[labell]{$v_7^{i+1}$};
        \node (h14) at ({3*cos(14*(360/14) + (360/28)) + 0.2}, 1.4)[labell]{$v_8^{i+1}$};

        \draw[red] (i3) to[bend right=5] (i10);
        \draw[red] (i3) to (i12);
        \draw[red] (i3) to (i8);
        \draw[red] (i2) to (i12);
        \draw[red] (i2) to (i13);
        \draw[red] (i8) to (i4);
        \draw[red] (i4) to (i7);
        \draw (i6) to (i7) ;
        \draw (i13) to (i14) ;
        \draw (i1) to (i2) ;
        \draw (i4) to (i5) ;
        \draw (i8) to (i9) ;
        \draw (i11) to (i12) ;    
        \draw (i3) to[bend left=5] (i10);
        \draw (i10) to (i1);
        \draw (i10) to (i5);
        \draw (i14) to (i10);
        \draw (i11) to (i14);
        \draw (i6) to (i10);
        \draw (i9) to (i6);
        \draw[red] (i6) to[bend left = 20] (i7) ;
        \draw[red] (i13) to[bend left = 20] (i14) ;
        \draw[red] (i1) to[bend left = 20] (i2) ;
        \draw[red] (i4) to[bend left = 20] (i5) ;
        \draw[red] (i8) to[bend left = 20] (i9) ;
        \draw[red] (i11) to[bend left = 20] (i12) ;

        \draw[red, dashed] (i7) to ({3*cos(7*(360/14) + (360/28)) - 0.3}, -0.4);
        \draw[dashed] (i6) to ({3*cos(6*(360/14) + (360/28)) - 0.3}, 0.4);
        \draw[red, dashed] (i13) to ({3*cos(13*(360/14) + (360/28)) + 0.3}, -0.4);
        \draw[dashed] (i14) to ({3*cos(14*(360/14) + (360/28)) + 0.3}, 0.4);
            
        \begin{scope}
            \clip (0,0) ellipse (3cm and 2cm);
        \end{scope}

        \end{tikzpicture}
        \end{center}
        \caption{On the left, the tree $T_1$ in black and the tree $T_1'$ in red. On the right, the subgraph induced by $C_i$ and $C_{i+1}$ in a pair $T_k$ (in black) and $T_k'$ (in red). }
        \label{fig:lower}
    \end{figure}
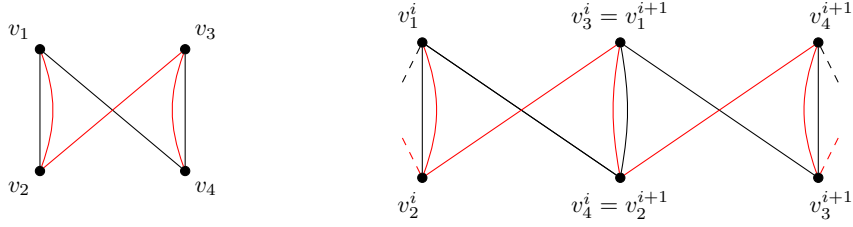

\paragraph{Construction of the trees.}

Let us denote by $T_1$ and $T'_1$ the pair of non-crossing spanning trees on a convex set $C$ of size $8$ represented in Figure~\ref{fig:lower}. Note that we have $\delta(T_1,T_1')=3$.

For every $k$, we denote by $T_k, T_k'$ the pair of non-crossing spanning trees obtained by taking $k$ disjoint copies of $T_1,T_1'$ and identifying the points $v_7$ and $v_8$ of the $i$-th copy  respectively with the points $v_2$ and $v_1$ in the $(i+1)$-th copy for $i<k$. (Note that the identification is performed upside down, which will be of importance in the proof, see Figure~\ref{fig:lower} for an illustration with $i=2$). 
We again define $C_i$ as the set of points of the $i$-th copy, and $v_j^i$ the point corresponding to $v_j$ in $C_i$. 
Observe that $\delta(T_k, T_k') = 3k$ for all $k \geq 1$.

\medskip

\paragraph{Properties of a minimal flip sequence.}

We first claim that for every $k\geq 1$, there is a flip sequence from $T_k$ to $T_k'$ of length $\frac{5}{3}\delta(T_k, T_k') = 5k$.
Indeed, the following flip sequence gives a transformation from $T_1$ to $T_1'$: we perform in order the flips $v_6v_1 \rightsquigarrow v_2v_5$, $v_3v_8 \rightsquigarrow v_4v_7$, $v_1v_8 \rightsquigarrow v_4v_5$, $v_2v_5 \rightsquigarrow v_2v_4$, and finally $v_4v_7 \rightsquigarrow v_5v_7$.
We can adapt this flip sequence for every $k>1$ between $T_k$ and $T'_k$ into a sequence of length $5k$ by applying the former in each copy of $T_1$ and $T_1'$ independently. 
The rest of the proof of Theorem~\ref{thm:lb_flips} consists in proving by induction on $k$ that the above mentioned sequences are minimal.
First, we prove the base case $k=1$.

 \begin{lemma}\label{cl:T1flip}
    A minimal flip sequence between $T_1$ and $T'_1$ has length at least $5$.
 \end{lemma}
 \begin{proof}
    Since every chord of $T_1$ crosses all chords of $T_1'$, the first two flips cannot create a chord of $T_1'$.
    Thus, after the first two steps, the symmetric difference still contains at least three edges of $T_1'$.
    Hence, a flip sequence between $T_1$ and $T'_1$ has length at least $5 = \frac{5}{3}\delta(T_1, T_1')$.
 \end{proof}

    Let $k>1$ be such that for $\ell < k$, a minimal flip sequence from $T_\ell$ to $T_\ell'$ has length at least $5\ell$. Following the exact same arguments as in Lemma~\ref{lem:LB_ncflips_chordschange}, we can derive the following.
    
\begin{lemma}\label{lem:LB_flips_chordschange}
    If a flip sequence $\S$ from $T_k$ to $T_k'$ does not modify at least one common chord $e$, then $\S$ has length at least $5k$.
\end{lemma}

Therefore, it only remains to show that flip sequences that modify every common chords of $T_k$ and $T'_k$ have length at least $5k$. Let $\S$ be such a sequence. We use a more involved version of the counting argument presented in Lemma~\ref{lem:LB_ncflips_nochordschange}. More precisely, we distribute one unit of weight to a subset of $C_1, \dots C_k$ for every flip of $\S$. We will essentially\footnote{What we will prove is actually slightly weaker since we will only ensure that the total weight is at least $5k-3/4$; But since the weight has to be an integer, it will be enough to conclude.} show that the total weight given by $\S$ to every set $C_i$ is at least $5$, which ensures that $\S$ has length at least $5k$. In other words, we will prove that the following holds:
    
    \begin{lemma}\label{lem:flipcounting}
        Let $\S$ be a minimal flip sequence between $T_k$ and $T_k'$ such that all the common chords are modified. Then, $\S$ has length at least $5k$.
    \end{lemma}

\paragraph{Assignment of weights.} 
    Let $\S$ be a minimal flip sequence from $T_k$ to $T_k'$ modifying all common chords.
    Recall that during a step, one edge is added and one is removed.
    We call the addition of an edge or the removal of an edge a \emph{phase} of the flip sequence (each step then consists of two phases).
    
    We distribute $\frac{1}{2}$ for each phase as follows.
    Consider a phase of a step in $\S$ in which an edge $e = uv$ is added or removed.
    Let $i$ and $j$ be the minimal indices such that $u \in C_i$ and $v \in C_j$, with $i \le j$ by symmetry.
    The sequence $\S$ gives a total weight of $\frac{1}{2}$ for this phase according to the following rules:
    \begin{itemize}
        \item if $u$ also belongs to $C_{i+1}$ and $v \in \bigcup_{\ell>i}C_{\ell}$, $\S$ gives weight $\frac{1}{4}$ from $u$ to $C_{i+1}$, otherwise, $\S$ gives weight $\frac{1}{4}$ from $u$ to $C_i$.
        \item $\S$ gives weight $\frac{1}{4}$ from $v$ to $C_j$. 
    \end{itemize}
    
    Note that, if $uv$ is a common chord of $T_k$ and $T'_k$, then $i = j$ and both $C_i$ and $C_{i+1}$ receive weight $\frac{1}{4}$.
    
    The rest of the proof consists in  counting in different claims how much weight is given during all the phases of $\S$.

\paragraph{Weights assignment for disjoint and common chords.}
    
    \begin{claim}\label{cl:flip3chords}
        For every $i$, $\S$ gives weight $3$ to $C_i$ because of the addition of the chords of $T_k' \setminus T_k$ and the deletion of the chords of $T_k \setminus T_k'$ with both endpoints in $C_i$.
    \end{claim}
        \begin{proof}
        Each set $C_i$ receives weight $\frac{1}{2}$ when a chord of $T_k \setminus T_k'$ with both endpoints in $C_i$ is deleted, and each set $C_i$ receives the same amount when a chord of $T'_k \setminus T_k$ with both endpoints in $C_i$ is added.
        Since the chords of $T_k' \setminus T_k$ have to be added and the chords of $T_k \setminus T_k'$ have to be deleted during $S$ and there are six such chords, 
        the conclusion follows.
        \end{proof}

    \begin{claim}\label{cl:flip1common}
    For every $1<i<k$, $\S$ gives weight $1$ to $C_i$ because of the addition and deletion of the common chords of $T_k$ and $T_k'$. The sets $C_1$ and $C_k$ receive weight $\frac 12$.
    \end{claim}
        \begin{proof}
            Recall that we suppose that the minimal flip sequence $\S$ changes all the common chords.
            Each set $C_i$ receives weight $\frac{1}{4}$ when a chord in $T_k \cap T'_k$ with both endpoints in $C_i$ is deleted, and $C_i$ receives the same amount when such a chord is added back.
            Since $\S$ modifies all common chords and there are two such chords for $1 < i <k$ and one otherwise, the conclusion follows.
        \end{proof}

    Claim~\ref{cl:flip3chords} and Claim~\ref{cl:flip1common} ensure that $\S$ already gives a total weight of at least $4k - 1$. We now aim at finding $k+1$ additional units of weight. This weight can only come from phases of $\S$ involving intermediate edges, that are edges not appearing in $T_k$ nor $T'_k$.

\paragraph{Weight assignment for intermediate edges.}
The core of the proof consists in proving the following claim:

\begin{claim}\label{cl:flip2connect}
    For each $i$, there exist two distinct edges $e_1$, $e_2$  that are not in $T_k \cup T_k'$, such that $\S$ gives weight $1$ to $C_i$ because of the addition and the deletion of $e_1$ and $e_2$.
    Moreover, $C_i$ receives this weight from two endpoints $u_1 \in e_1$ and $u_2 \in e_2$ which are not from $v_2^i$ and $v_7^i$ for any $i \le k$.
\end{claim}
    \begin{proof}
        Recall that, since all the edges of $T_k \setminus T_k'$ with both endpoints in $C_i$ are pairwise crossing the edges of $T_k' \setminus T_k$ with both endpoints in $C_i$, we have to modify all the edges of $T_k \setminus T_k'$ with both endpoints in $C_i$ before creating an edge of $T_k' \setminus T_k$ with both endpoints in $C_i$. 
        Consider the last step where a chord $e$ of $T_k \setminus T_k'$ with both endpoints in $C_i$ is removed during $\S$. And let us denote by $T$ the tree before removing $e$. Note that $e$ is the only chord of $T \cap \Delta(T_k,T'_k)$.
        
        We want to prove that there are two different edges $e_1$ and $e_2$ in $T \setminus (T_k' \cup T_k)$, each having an endpoint different from $v_2^i$ and $v_7^i$ that gives weight $\frac{1}{4}$ when adding and removing $e_1$ and $e_2$. We distinguish two cases (see Figure~\ref{fig:weightflip} for an illustration). \smallskip

        \noindent\textbf{Case 1:} $e = v_1^iv_8^i$. \\
        By connectivity, there exists an edge $e_1$ (resp. $e_2$) in $T \setminus (T_k \cup T_k')$ with exactly one endpoint in $\{ v_5^i, v_6^i \}$ (resp. $\{ v_3^i, v_4^i \}$) in $T$. Note that $e_1, e_2$ and $v_2^iv_7^i$ are pairwise distinct since $e$ separates $\{v_3^i,v_4^i,v_7^i\}$ from $\{v_2^i,v_5^i,v_6^i\}$.
        
        \smallskip
        
        \noindent\textbf{Case 2:} $e$ is either $v_1^iv_6^i$ or $v_3^iv_8^i$. \\
        Up to symmetry, we can assume that $e=v_3^iv_8^i$.

         Let $A$ be the  side of $e=v_3^iv_8^i$ containing $v_1^i$. By connectivity, there exist two edges $e_1$ (resp. $e_2$) in $T\setminus (T_k \cup T_k')$ with exactly one endpoint in $\{v_3^i,v_8^i \}$  (resp. $\{ v_5^i, v_6^i \}$) and the other endpoint in $A$. If $e_1 \ne e_2$, the conclusion follows. Otherwise, there exists another edge $e'_2$ in $T \setminus (T_k \cup T_k')$ with exactly one endpoint in $\{ v_3^i,v_5^i,v_6^i,v_8^i \}$ and the other endpoint in $A$, which completes the proof.

        \smallskip
        
        In both cases, we have proved the existence of the two distinct edges $e_1$ and $e_2$ in $T \setminus (T_k' \cup T_k)$.
        Since $e_1$ and $e_2$ have to be both created and removed in $\S$, $\S$ gives weight $1$ to $C_i$ because of the addition and deletion of $e_1$ and $e_2$. 
    \end{proof}

    \begin{figure}[hbtp]
        \begin{center}
        \tikzstyle{vertex}=[circle,draw, minimum size=7pt, scale=0.5, inner sep=1pt, fill = black]
        \tikzstyle{fleche}=[->,>=latex] 
        \tikzstyle{labell}=[text opacity=1, scale =1]
        \begin{tikzpicture}[scale=0.9]

            \node (a1) at (-110:2) [vertex] {};
            \node (a2) at ({2*cos(160)}, -1.2) [vertex] {};
            \node (a3) at (-70:2) [vertex] {};
            \node (a4) at ({2*cos(160)}, 1.2) [vertex] {};
            \node (a5) at ({2*cos(20)}, -1.2) [vertex] {};
            \node (a6) at (110:2) [vertex] {};
            \node (a7) at ({2*cos(20)}, 1.2) [vertex] {};
            \node (a8) at (70:2) [vertex] {};

            \node (b1) at (-110:2.4) [labell] {$v_5^i$};
            \node (b2) at ({2.4*cos(160)}, -1.2) [labell] {$v_2^i$};
            \node (b3) at (-70:2.4) [labell] {$v_6^i$};
            \node (b4) at ({2.4*cos(160)}, 1.2) [labell] {$v_1^i$};
            \node (b5) at ({2.4*cos(20)}, -1.2) [labell] {$v_8^i$};
            \node (b6) at (110:2.4) [labell] {$v_3^i$};
            \node (b7) at ({2.4*cos(20)}, 1.2) [labell] {$v_7^i$};
            \node (b8) at (70:2.4) [labell] {$v_4^i$};
            
            \draw[opacity=0.55, densely dashed] (a1) to (a3);
            \draw[opacity=0.55, densely dashed] (a2) to (a4);
            \draw[opacity=0.55, densely dashed] (a6) to (a8);
            \draw[opacity=0.55, densely dashed] (a5) to (a7);
            \draw[opacity=0.55, densely dashed] (a5) to (a6);
            \draw[opacity=0.55, densely dashed] (a5) to (a4);
            \draw[opacity=0.55, densely dashed] (a3) to (a4);

            \draw[opacity=0.55, densely dashed] (a2) to [bend right = 20](a4) [red];
            \draw[opacity=0.55, densely dashed] (a7) to [bend right = 20](a5) [red];

            \draw[opacity=0.55, densely dashed] (a1) to (a7) [red];
            \draw[opacity=0.55, densely dashed] (a2) to (a8) [red];
            \draw[opacity=0.55, densely dashed] (a1) to (a8) [red];

            \draw[thick, blue] (a4) to (a5);
            \draw[thick, blue] (a1) to (a3);
            \draw[thick, blue] (a6) to (a8);
            \draw[thick, blue] (a1) to (a5);
            \draw[thick, blue] (a8) to[bend right = 50] (2,0.3);

            \node (b8) at (1.6,-1.7) [labell, blue, scale=1.2] {$e_1$};
            \node (b8) at (1.3, 1.3) [labell, blue, scale=1.2] {$e_2$};
            \node (b8) at (-1.2, 1.2) [labell, blue, scale=1.2] {$e$};
            
        \tikzset{xshift=7.5cm}
     
            \node (a1) at (-110:2) [vertex] {};
            \node (a2) at ({2*cos(160)}, -1.2) [vertex] {};
            \node (a3) at (-70:2) [vertex] {};
            \node (a4) at ({2*cos(160)}, 1.2) [vertex] {};
            \node (a5) at ({2*cos(20)}, -1.2) [vertex] {};
            \node (a6) at (110:2) [vertex] {};
            \node (a7) at ({2*cos(20)}, 1.2) [vertex] {};
            \node (a8) at (70:2) [vertex] {};

            \node (b1) at (-110:2.4) [labell] {$v_5^i$};
            \node (b2) at ({2.4*cos(160)}, -1.2) [labell] {$v_2^i$};
            \node (b3) at (-70:2.4) [labell] {$v_6^i$};
            \node (b4) at ({2.4*cos(160)}, 1.2) [labell] {$v_1^i$};
            \node (b5) at ({2.4*cos(20)}, -1.2) [labell] {$v_8^i$};
            \node (b6) at (110:2.4) [labell] {$v_3^i$};
            \node (b7) at ({2.4*cos(20)}, 1.2) [labell] {$v_7^i$};
            \node (b8) at (70:2.4) [labell] {$v_4^i$};

            \draw[opacity=0.55, densely dashed] (a1) to (a3);
            \draw[opacity=0.55, densely dashed] (a2) to (a4);
            \draw[opacity=0.55, densely dashed] (a6) to (a8);
            \draw[opacity=0.55, densely dashed] (a5) to (a7);
            \draw[opacity=0.55, densely dashed] (a5) to (a6);
            \draw[opacity=0.55, densely dashed] (a5) to (a4);
            \draw[opacity=0.55, densely dashed] (a3) to (a4);

            \draw[opacity=0.55, densely dashed] (a2) to [bend right = 20](a4) [red];
            \draw[opacity=0.55, densely dashed] (a7) to [bend right = 20](a5) [red];

            \draw[opacity=0.55, densely dashed] (a1) to (a7) [red];
            \draw[opacity=0.55, densely dashed] (a2) to (a8) [red];
            \draw[opacity=0.55, densely dashed] (a1) to (a8) [red];

            \draw[thick, blue] (a6) to (a5);
            \draw[thick, blue] (a1) to (a3);
            \draw[thick, blue] (a6) to (a8);
            \draw[thick, blue] (a6) to [bend left = 45] (-2,0.3);
            \draw[thick, blue] (a3) to [bend right = 35] (-2,-0.3);

            \node (b8) at (-0.8,-0.8) [labell, blue, scale=1.2] {$e_2$};
            \node (b8) at (-1.2, 1.2) [labell, blue, scale=1.2] {$e_1$};
            \node (b8) at (0.7, 0.7) [labell, blue, scale=1.2] {$e$};

        \end{tikzpicture}
        \end{center}
        \caption{The tree $T$ (in blue) on the set $C_i$ obtained before removing $e$ during $\S$ in the Claim~\ref{cl:flip2connect}. On the left, $e= v_1^iv_8^i$, and on the right, $e = v_8^iv_3^i$. In both cases, the two edges $e_1$ and $e_2$ are not in $T_k \cup T_k'$, and $C_i$ receives weight from each phase adding or removing them. }
        \label{fig:weightflip}
    \end{figure}
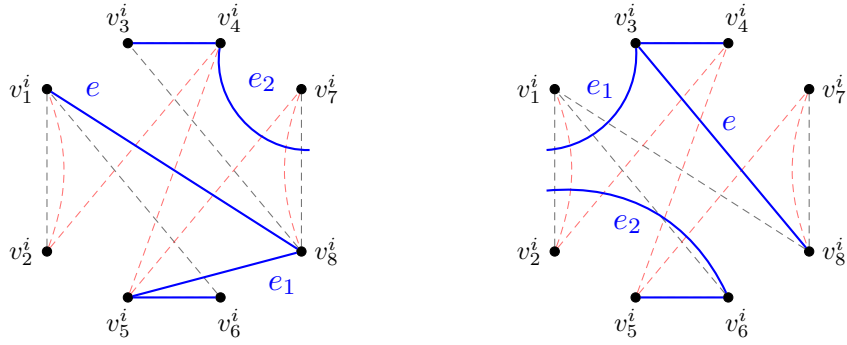

All the claims above put together ensure that the total weight given by $\S$ over all phases is at least $5k - 1$.
Since the total weight is an integer, in order to ensuring that the flip sequence has length at least $5k$, we only need to find some positive additional weight given by $\S$.

\begin{claim}\label{cl:fliplast}
    The flip sequence $\S$ gives an additional weight of $\frac{1}{4}$ to some set $C_i$.
\end{claim}
    \begin{proof}
    Let $T$ be the first tree of $\S$ where an edge $e_0$ of $T_k \cap T_k'$ has been removed. Let us denote by $i$ the index such that $e_0$ is in both $C_i$ and $C_{i+1}$.
    Then, there is an edge $e^* = u^*w^* \neq e_0$ in $T$ such that $u^* = v_7^i = v_2^{i+1}$.
    Since $e^* \neq e_0$ and $v_7^i$ is an endpoint of $e^*$, $e^*$ is not in $T_k$.
    Observe that $\S$ gives weight $\frac{1}{4}$ to either $C_i$ or $C_{i+1}$ from $v_7^i$ when adding $e^*$. This weight has not been counted in Claim~\ref{cl:flip2connect} (by assumption on $u_1$ and $u_2$). If $e^*$ is not in $T_k'$,  then this weight was not counted by Claim~\ref{cl:flip1common} either and we are done. 
    
    Otherwise, $e^*$ is either $v_2^{i+1}v_4^{i+1}$ or $v_5^iv_7^i$, say the latter by symmetry. 
    First note that the chords of $T_k \setminus T_k'$ with both endpoints in $C_i$ are not in $T$ (since $e^*$ crosses all these edges).
    By connectivity, since no common chord has been removed in $\S$ before $e_0$ and since $T$ does not contain any chord of $T_k$, at least three edges of $T \setminus T_k$ distinct from $e^*$ have both endpoints in $C_i$.
    In particular, at least one of them, denoted by $e'$, is not in $T_k'$.
    Since $e'$ is not in $T_k \cup T_k'$ and has both endpoints in $C_i$, the weight given by $\S$ to $C_i$ when removing $e'$ was not entirely counted by Claim~\ref{cl:flip2connect}. Indeed, this claim only considers the contribution of exactly one endpoint of each edge $e_1$ and $e_2$, not both.
    \end{proof}

All the previous claims ensure that the weight given by $\S$ is at least $\lceil 5k- \frac{3}{4}  \rceil = 5k$, which completes the proof of Lemma~\ref{lem:flipcounting}.

\subsection{Rotations}\label{sec:lowerrotation}

In this section, we also give a family $(T_k, T_k')_{k \in \NN^*}$ that satisfies the conclusion of Theorem~\ref{thm:lb_rotations}.

\lbrotations*

We consider the same inductive construction as before, but start with a slightly different pair $(T_1, T_1')$: the edge $v_1v_8$ is replaced by $v_3v_6$ in $T_1$, as illustrated in Figure~\ref{fig:lower2} (the graph $T_1'$ remaining the same).

\begin{figure}[hbtp]
        \begin{center}
        \tikzstyle{vertex}=[circle,draw, minimum size=7pt, scale=0.5, inner sep=1pt, fill = black]
        \tikzstyle{fleche}=[->,>=latex] 
        \tikzstyle{labell}=[text opacity=1, scale =0.85]
        \begin{tikzpicture}[scale=0.9]

            \node (a1) at (-110:2) [vertex] {};
            \node (a2) at (-160:2) [vertex] {};
            \node (a3) at (-70:2) [vertex] {};
            \node (a4) at (160:2) [vertex] {};
            \node (a5) at (-20:2) [vertex] {};
            \node (a6) at (110:2) [vertex] {};
            \node (a7) at (20:2) [vertex] {};
            \node (a8) at (70:2) [vertex] {};

            \node (b1) at (-110:2.4) [labell] {$v_5$};
            \node (b2) at (-160:2.4) [labell] {$v_2$};
            \node (b3) at (-70:2.4) [labell] {$v_6$};
            \node (b4) at (160:2.4) [labell] {$v_1$};
            \node (b5) at (-20:2.4) [labell] {$v_8$};
            \node (b6) at (110:2.4) [labell] {$v_3$};
            \node (b7) at (20:2.4) [labell] {$v_7$};
            \node (b8) at (70:2.4) [labell] {$v_4$};
            
            \draw (a1) to (a3);
            \draw (a2) to (a4);
            \draw (a6) to (a8);
            \draw (a5) to (a7);
            \draw (a5) to (a6);
            \draw (a3) to (a6);
            \draw (a3) to (a4);

            \draw (a1) to [bend left = 20](a3) [red];
            \draw (a2) to [bend right = 20](a4) [red];
            \draw (a7) to [bend right = 20](a5) [red];
            \draw (a8) to [bend left = 20](a6) [red];

            \draw (a1) to (a7) [red];
            \draw (a2) to (a8) [red];
            \draw (a1) to (a8) [red];

        \tikzset{xshift=7.5cm}
        

        \foreach \s in {1,...,5}{(\node (i\s) at ({3*cos(\s*(360/14) + (360/28))}, 2)[vertex]{};};
        \node (i6) at ({3*cos(6*(360/14) + (360/28))}, 1)[vertex]{};
        \node (i7) at ({3*cos(7*(360/14) + (360/28))}, -1)[vertex]{};
        \node (i13) at ({3*cos(13*(360/14) + (360/28))}, -1)[vertex]{};
        \node (i14) at ({3*cos(14*(360/14) + (360/28))}, 1)[vertex]{};
        \foreach \s in {8,...,12}{(\node (i\s) at ({3*cos(\s*(360/14) + (360/28))}, -2)[vertex]{};};

        \node (h1) at ({3*cos(1*(360/14) + (360/28))}, 2.4)[labell]{$v_6^{i+1}$};
        \node (h2) at ({3*cos(2*(360/14) + (360/28))}, 2.4)[labell]{$v_5^{i+1}$};
        \node (h3) at ({3*cos(3*(360/14) + (360/28))}, 2.4)[labell]{$v_7^i = v_2^{i+1}$};
        \node (h4) at ({3*cos(4*(360/14) + (360/28))}, 2.4)[labell]{$v_4^i$};
        \node (h5) at ({3*cos(5*(360/14) + (360/28))}, 2.4)[labell]{$v_3^i$};
        \node (h8) at ({3*cos(8*(360/14) + (360/28))}, -2.4)[labell]{$v_5^i$};
        \node (h9) at ({3*cos(9*(360/14) + (360/28))}, -2.4)[labell]{$v_6^i$};
        \node (h10) at ({3*cos(10*(360/14) + (360/28))}, -2.4)[labell]{$v_8^i = v_1^{i+1}$};
        \node (h11) at ({3*cos(11*(360/14) + (360/28))}, -2.4)[labell]{$v_3^{i+1}$};
        \node (h12) at ({3*cos(12*(360/14) + (360/28))}, -2.4)[labell]{$v_4^{i+1}$};
        \node (h6) at ({3*cos(6*(360/14) + (360/28)) - 0.2}, 1.4)[labell]{$v_1^i$};
        \node (h7) at ({3*cos(7*(360/14) + (360/28)) - 0.2}, -1.4)[labell]{$v_2^i$};
        \node (h13) at ({3*cos(13*(360/14) + (360/28)) + 0.2}, -1.4)[labell]{$v_7^{i+1}$};
        \node (h14) at ({3*cos(14*(360/14) + (360/28)) + 0.2}, 1.4)[labell]{$v_8^{i+1}$};

        \draw[red] (i3) to[bend right=5] (i10);
        \draw[red] (i3) to (i12);
        \draw[red] (i3) to (i8);
        \draw[red] (i2) to (i12);
        \draw[red] (i2) to (i13);
        \draw[red] (i8) to (i4);
        \draw[red] (i4) to (i7);
        \draw (i6) to (i7) ;
        \draw (i13) to (i14) ;
        \draw (i1) to (i2) ;
        \draw (i4) to (i5) ;
        \draw (i8) to (i9) ;
        \draw (i11) to (i12) ;    
        \draw (i3) to[bend left=5] (i10);
        \draw (i10) to (i1);
        \draw (i10) to (i5);
        \draw (i1) to (i11);
        \draw (i11) to (i14);
        \draw (i5) to (i9);
        \draw (i9) to (i6);
        \draw[red] (i6) to[bend left = 20] (i7) ;
        \draw[red] (i13) to[bend left = 20] (i14) ;
        \draw[red] (i1) to[bend left = 20] (i2) ;
        \draw[red] (i4) to[bend left = 20] (i5) ;
        \draw[red] (i8) to[bend left = 20] (i9) ;
        \draw[red] (i11) to[bend left = 20] (i12) ;

        \draw[red, dashed] (i7) to ({3*cos(7*(360/14) + (360/28)) - 0.3}, -0.4);
        \draw[dashed] (i6) to ({3*cos(6*(360/14) + (360/28)) - 0.3}, 0.4);
        \draw[red, dashed] (i13) to ({3*cos(13*(360/14) + (360/28)) + 0.3}, -0.4);
        \draw[dashed] (i14) to ({3*cos(14*(360/14) + (360/28)) + 0.3}, 0.4);
            
        \begin{scope}
            \clip (0,0) ellipse (3cm and 2cm);
        \end{scope}

        \end{tikzpicture}
        \end{center}
        \caption{On the left, the tree $T_1$ in black and the tree $T_1'$ in red. On the right, the subgraph induced by $C_i$ and $C_{i+1}$ in a pair $T_k$ (in black) and $T_k'$ (in red). }
        \label{fig:lower2}
    \end{figure}

For every $k \geq 1$, there is a rotation sequence from $T_k$ to $T_k'$ of length $7k$. 
Indeed, the following rotation sequence gives a transformation from $T_1$ to $T_1'$: $v_1v_6 \rightsquigarrow v_1v_5$, $v_3v_6 \rightsquigarrow v_3v_5$, $v_3v_5 \rightsquigarrow v_5v_8$, $v_5v_8 \rightsquigarrow v_5v_7$, $v_3v_8 \rightsquigarrow v_4v_5$, $v_1v_5 \rightsquigarrow v_1v_4$ and finally $v_1v_4 \rightsquigarrow v_2v_4$.
This rotation sequence can indeed be generalized into a rotation sequence between $T_k$ and $T'_k$ of length $7k$ by rotating in each copy of $T_1$ and $T_1'$ independently.
The rest of the proof consists in proving by induction on $k$ that such sequences are minimal. 

For the base case, one can check that seven rotations are needed to transform $T_1$ into $T'_1$ similarly to the previous sections, but the case analysis is quite tedious. We rather run an exhaustive computer search\footnote{The code can be found at \href{https://github.com/tpierron/reconf-nc-trees}{https://github.com/tpierron/reconf-nc-trees}} which checks that $7$ is indeed the length of a minimal rotation sequence. 

\begin{lemma}\label{cl:T1rotation}
    A minimal rotation sequence between $T_1$ and $T'_1$ has length at least $7$.
\end{lemma}

Assume now that $k > 1$, and, for every $\ell < k$, any rotation sequence between $T_\ell$ and $T_\ell'$ has length at least $7\ell$. 
Let $\S$ be a minimal rotation sequence from $T_k$ to $T_k'$. 

Following the steps of the previous parts, we may assume that $\S$ modifies every common chord of $T_k$ and $T_k'$, for otherwise we can mimic the proof of Lemma~\ref{lem:LB_flips_chordschange} and directly conclude by induction.

The rest of the proof is different from the proofs of the previous sections. While we only counted edges created or removed during the sequence in previous sections and proved that this number is large, we use a more involved argument here consisting in proving that the number of rotations is large.
We start with an easy claim about the number of rotations that involve the edges of $\Delta(T_k,T'_k)$. 
\begin{claim}
\label{cl:rotate3chords}
 The following subset $\S_0$ of $\S$ gives $6k$ pairwise disjoint rotations:
\begin{itemize}
    \item for each edge of $T_k\setminus T'_k$, $\S_0$ contains the first rotation that removes it and,
    \item for each edge of $T'_k\setminus T_k$, $\S_0$ contains the last rotation that creates it.
\end{itemize} 
\end{claim}
\begin{proof}    
Note that the chords of $T_k \setminus T_k'$ have no common endpoint with the chords of $T_k' \setminus T_k$, hence every rotation that removes a chord of $T_k  \setminus T_k'$ cannot create a chord of $T_k' \setminus T_k$. Therefore $\S_0$ contains $6k$ rotations, which concludes.
\end{proof}
    
It remains to prove that there are $k$ rotations that have not yet been counted, \emph{i.e.} $k$ steps in $\S\setminus \S_0$. First, we prove the existence of rotations in $\S$ involving $C_i$ and $C_{i+1}$ which have special properties.

\begin{lemma}\label{cl:rotspecial} 
    For each $i\in[1,k-1]$, let $V_i :=( C_i \cup C_{i+1} ) \setminus \{ v_1^i, v_2^i, v_7^{i+1}, v_8^{i+1} \}$.
    The sequence $\S$ contains at least one of the following: 
    \begin{itemize}
    \item[(1)] a rotation $v_7^iv_8^i\rightsquigarrow v_5^iv_7^i$ (or $v_7^iv_8^i\rightsquigarrow v_2^{i+1}v_4^{i+1}$ by symmetry) not in $\S_0$,
    \item[(2)] a rotation $v_3^iv_8^i\rightsquigarrow v_7^iv_8^i$ (or $v_1^{i+1}v_6^{i+1}\rightsquigarrow v_7^iv_8^i$ by symmetry) not in $\S_0$,
    \item[(3)] a rotation that removes $v_7^iv_8^i$ to add an edge that is not in $T_k' \setminus T_k$,
    \item[(4)] a rotation that removes an edge that is not in $T_k \setminus T_k'$ to add $v_7^iv_8^i$,
    \item[(5)] two rotations, one that removes $v_8^iu$ with $u \in \{v_5^i, v_6^i\}$ (or $v_1^{i+1}u$ with $u \in \{v_3^{i+1}, v_4^{i+1}\}$ by symmetry) to add an edge that is not in $T_k' \setminus T_k$ with both endpoints in $V_i$, and the other that removes an edge that is not in $T_k \setminus T_k'$ with both endpoints in $V_i$ to add $v_2^{i+1}w$ with $w \in \{v_5^{i+1},v_6^{i+1}\}$ (or $v_7^iw$ with $w \in \{v_3^{i}, v_4^{i}\}$ by symmetry). 
    \end{itemize} 
\end{lemma}

\begin{proof}   
    Assume by contradiction that $\S$ does not contain such a rotation for an integer $i$ and let us denote by $e$ the edge $v_7^iv_8^i=v_1^{i+1}v_2^{i+1}$ (which is in $C_i\cap C_{i+1}$). 
    Recall that all the common chords of $T_k$ and $T_k'$ are rotated during the sequence. 
    Let $T$ be the first tree in $\S$ that does not contain $e$ and $e \rightsquigarrow e'$ be the roation applied to obtain $T$.
    Since (3) does not hold, $e'$ must be an edge of $T_k'$. 
    Thus, $e'$ is either $v_2^{i+1}v_4^{i+1}$ or $v_5^iv_7^i$, say $v_5^iv_7^i$ by symmetry. 
    Since (1) does not hold, the flip $e\rightsquigarrow e'$ is in $\S_0$, hence $v_5^iv_7^i$ is in every tree obtained after $T$ during $\S$.

    Let us prove that $e$ is rotated exactly once.
    Assume by contradiction that it is rotated a second time (after being added back) into an edge $f$ that is in $T_k' \setminus T_k$ since (3) does not hold. Since (1) does not hold and because of the existence of $e'$, $f=v_2^{i+1}v_4^{i+1}$, and every tree obtained afterwards contains $v_5^iv_7^i$ and $v_2^{i+1}v_4^{i+1}$.
    Since $e\in T'_k$, there is an edge $e''$ rotated into $e$ after creating $v_2^{i+1}v_4^{i+1}$, and let $T''$ be the tree obtained from $T_k$ just before this rotation takes place. Since $v_5^iv_7^i$ and $v_2^{i+1}v_4^{i+1}$ are in $T''$, $e''$ is neither $v_3^iv_8^i$ nor $v_1^{i+1}v_6^{i+1}$. So, $e'' \notin T_k \setminus T_k'$, which contradicts (4).
    So, from now on, we will assume that the edge $e$ is removed and added exactly once.
    
    Let $T'$ be the tree in $\S$ before adding back $e$ and $e^* \rightsquigarrow e$ be the rotation applied at that step.
    Since (4) does not hold, $e^* \in T_k\setminus T'_k$, hence $e^*$ is either $v_3^iv_8^i$ or $v_1^{i+1}v_6^{i+1}$. Since $T'$ appears after $T$ in $\S$, $v_5^iv_7^i$ is in $T'$ and then $T'$ cannot contain $v_3^iv_8^i$, and then $e^* = v_1^{i+1}v_6^{i+1}$.
    Since (2) does not hold, $e^*\rightsquigarrow e$ is in $\S_0$ and $e^* = v_1^{i+1}v_6^{i+1}$ belongs to all the trees before $T'$ in $\S$.

    Since $v_1^{i+1}v_6^{i+1}$ and $e$ are in every tree obtained before $T$ during $\S$, and $e \rightsquigarrow v_5^iv_7^i$ is performed to obtain $T$, there is a path in $T$ connecting $\{v_1^{i+1}, v_6^{i+1} \}$ and $\{v_5^i, v_7^i \}$.
    Since this path is in the tree obtained before $T$, this path does not cross $v_7^iv_8^i$, and is not included in $C_{i+1}$.
    So the path is included in $C_{i+1}$ and $T$ contains an edge $v_8^iu$ such that $u \in \{v_5^i, v_6^i\}$.
    This edge is not in $T_k$ nor $T_k'$, thus it must be rotated during $\S$ after obtaining $T$. 
    Since (4) does not hold, it is not rotated into $e$.
    And, since $v_5^iv_7^i$ is in all the trees  after $T$ during $\S$ and is the only edge of $T_k' \setminus T_k$ that can share an endpoint with $v_8^iu$, the edge $v_8^iu$ cannot be rotated into an edge of $T_k' \setminus T_k$.
    Since $e$ is removed exactly once and (4) does not hold, there either $v_6^{i+1}v_1^{i+1}$ or $e$ belong to all the trees obtained after $T$.
    So $v_8^iu$ is rotated into an edge $b$ that is not in $T_k' \setminus T_k$, with both endpoints in $V_i$.
    
    Likewise, $v_5^iv_7^i$ and $e$ are in every tree obtained after $T'$ during $\S$, and $v_5^iv_7^i$ and $v_6^{i+1}v_1^{i+1}$ are in $T'$.
    So $T'$ contains an edge $v_2^{i+1}w$ such that $w \in \{v_5^{i+1}, v_6^{i+1}\}$.
    This edge is not in $T_k$ nor in $T'_k$, thus it must have been added during $\S$. 
    Since (3) does not hold, it has not been added by rotating $e$.
    And, since $v_6^{i+1}v_1^{i+1}$ is in each tree obtained before $T'$ during $\S$ and is the only edge of $T_k \setminus T_k'$ that can share an endpoint with $v_2^{i+1}w$, $v_2^{i+1}w$ cannot have been added by rotating an edge of $T_k \setminus T_k'$.
    Thus, $v_2^{i+1}w$ has been added by rotating an edge that is not in $T_k \setminus T_k'$, with both endpoints in $V_i$, which contradicts (5).
\end{proof}

Note that for each $i\in[1,k-1]$, several of the previous cases may arise. For $p\in[1,5]$, denote by $n_p$ the number of times case $(p)$ arises. The following lemma shows that all the rotations given by Lemma~\ref{cl:rotspecial} are pairwise distinct. This already gives $k-1$ rotations in $\S\setminus \S_0$, and even $k$ under the right conditions.

\begin{lemma}\label{cl:rotk-1}
    There are at least $n_1+n_2+n_3+n_4+2n_5$ rotations in $\S\setminus \S_0$.
\end{lemma}

\begin{proof}
We first observe that Lemma~\ref{cl:rotspecial} provides rotations not in $\S_0$. This is clear for items (1) and (2): they provide a rotation that does not lie in $\S_0$. Items (3) and (4) provide a rotation that does not involve edges in $\Delta(T_k,T'_k)$, thus not in $\S_0$. Finally, item (5) provides two rotations that also do not involve edges in $\Delta(T_k,T'_k)$, hence again not in $\S_0$.

Moreover, one can easily check that the rotations obtained applying Lemma~\ref{cl:rotspecial} to every $i\in[1,k-1]$ are pairwise distinct. This concludes, since items (1) to (4) each provide one rotation and item (5) provides two of them. 
\end{proof}

The previous lemmas ensure that a minimal rotation sequence $\S$ modifying all common chords contains at least $7k - 1$ rotations and if (5) happens at least once, then $\S$ has length at least $7k$. So from now on, we can assume that (5) never happens, and moreover for every $i\in[1,k-1]$, only one item among (1)-(4) happens. Denote by $r_i$ the corresponding rotation. Note that $r_i$ necessarily impacts the common edge $v_7^iv_8^i$, hence we say that $r_i$ (and by extension $v_7^iv_8^i$) has type $(p)$ when $r_i$ was provided by item $(p)$. It now remains to find a single additional rotation to conclude.

\begin{lemma}\label{lem:rotlastfarcopies}
    For every $i\in[1,k-1]$, each tree obtained during $\S$ contains at least one edge among $\{v_7^iv_8^i,v_3^iv_8^i,v_1^{i+1}v_6^{i+1}, v_5^iv_7^i,v_2^{i+1}v_4^{i+1}\}$. 
\end{lemma}

\begin{proof}
Observe that the chord $v_7^iv_8^i$ can only be deleted by a rotation in $\S_0$, or by $r_i$ (if the chord has type (1) or (3)). In the first case, the rotation creates $e'\in\{v_5^iv_7^i,v_2^{i+1}v_4^{i+1}\}$, and by definition of $\S_0$, the $e'$ is not deleted anymore afterwards. 

Assume that at some point, the edge $v_7^iv_8^i$ is deleted because of $r_i$. Then, $v_7^iv_8^i$ has to be created later, and it must be using a rotation from $\S_0$ (because $v_7^iv_8^i$ does not have type (2) nor (4)), therefore flipping $e''\in\{v_3^iv_8^i,v_1^{i+1}v_6^{i+1}\}$. By definition of $\S_0$, all the trees obtained before recreating $v_7^iv_8^i$ contain $e''$.  Now afterwards, all the rotations involving $v_7^iv_8^i$ are in $\S_0$, hence we can conclude using the previous case.
\end{proof}

\begin{lemma}\label{lem:rotlastnodouble}
    $\S\setminus \S_0$ contains a rotation that we have not already counted in Lemma~\ref{cl:rotk-1}. 
\end{lemma}

\begin{proof}

Assume by contradiction that every rotation from $\S\setminus\S_0$ has already been counted in Lemma~\ref{cl:rotk-1}.

For $i\in[1,k-1]$, let $e$ be the edge affected by $r_i$ different from the common chord $v_7^iv_8^i$. We say that $e$ \emph{interferes} with $C_i$ when the endpoint of $e$ not in $\{v_7^i,v_8^i\}$ is on the left of $v_7^iv_8^i$, and that it interferes with $C_{i+1}$ otherwise. Since there are $k-1$ common chords, there is an index $p$ such that no $r_i$ interferes with $C_p$.

\begin{claim}\label{cl:commonCp}
 There is no common chord that is rotated into $v_5^pv_7^p$ or into $v_2^pv_4^p$ and there is no common chord that has been added by rotating $v_3^pv_8^p$ or $v_1^pv_6^p$.   
\end{claim}
\begin{proof}
     Suppose that the common chord $e=v_7^pv_8^p$ is rotated into $v_5^pv_7^p$. 
    Let $T$ be the tree obtained before rotating $e$ into $v_5^pv_7^p$. 
    Since no rotations of type (1) interferes with $C_p$, the rotation $v_7^pv_8^p \rightsquigarrow v_5^pv_7^p$ is in $\S_0$ and $v_5^pv_7^p$ is in every tree obtained after $T$ during $\S$.
    By connectivity, there is an edge $xy \neq v_5^pv_7^p$ in $T$ such that $x \in \{ v_7^p, v_8^p\}$ and $y \in \bigcup_{j \leq i}C_j \setminus \{ v_7^p, v_8^p \}$.
    Since $v_5^pv_7^p$ can be added to $T \setminus \{ e\}$, $xy$ is not in $T_k \cup T_k'$.
    We distinguish two cases whether $x$ is $v_8^p$ or $v_7^p$. \smallskip
    
     \noindent\textbf{Case 1:} $x = v_8^p$. \\ 
    Then, $y \in \{ v_5^p, v_6^p \}$.
    So $xy$ is removed during $\S$ after obtaining $T$.
    Since $v_5^pv_7^p$ is not removed after obtaining $T$, there is a rotation that removes $xy$. This rotation is not in $\S_0$ since it cannot create an edge of $T_k' \setminus T_k$, and it was not counted by Lemma~\ref{cl:rotk-1} since otherwise it would interfere with $C_p$.  \smallskip
    
    \noindent\textbf{Case 2:} $x = v_7^p$. \\ 
    By definition of $p$, $xy$ has not been added by rotating a common chord with both endpoints in $C_p$. Moreover, by Lemma~\ref{lem:rotlastfarcopies}, $xy$ has not been added by rotating a common chord with an endpoint in $C_j$ with $j\leq p-2$.
    Therefore, we may assume that $xy$ has been added by a rotation of $\S_0$ (otherwise, we found an additional rotation). Since $xy\notin T'_k$, this rotation deleted a chord of $T_k \setminus T_k'$. In particular, since $x = v_7^p$ is not an endpoint of a chord of $T_k \setminus T_k'$, $y$ is an endpoint of a chord of $T_k \setminus T_k'$.

    Similarly, since $xy\notin T'_k$,  the edge $xy$ is removed during $\S$ after obtaining $T$. Applying the same argument, we get that $xy$ must be removed by a rotation from $\S_0$ that creates an edge $e'$ from $T'_k\setminus T_k$. Moreover, $e'$ is not $v_5^pv_7^p$ (since $v_5^pv_7^p$ is not removed after obtaining $T$) nor $v_1^{p+1}v_6^{p+1}$ (by connectivity).
    
    Thus, this rotation rotates around $y$, hence we also get that $y$ is an endpoint of a chord of $T_k' \setminus T_k$. This is a contradiction since chords of $T_k\setminus T'_k$ and of $T'_k\setminus T_k$ do not share any endpoint. \smallskip

    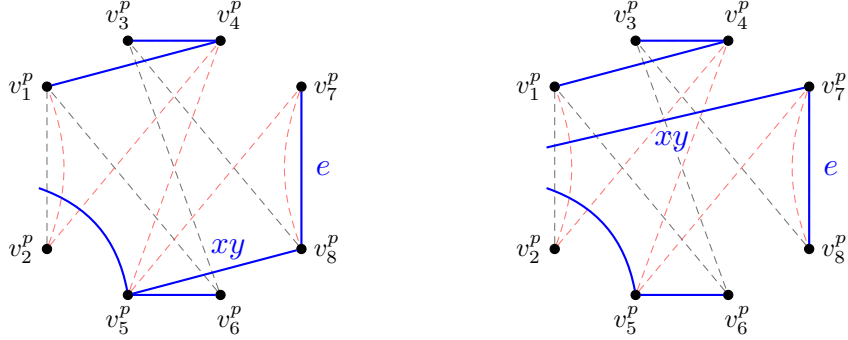
\begin{figure}[hbtp]
        \begin{center}
        \tikzstyle{vertex}=[circle,draw, minimum size=7pt, scale=0.5, inner sep=1pt, fill = black]
        \tikzstyle{fleche}=[->,>=latex] 
        \tikzstyle{labell}=[text opacity=1, scale =1]
        \begin{tikzpicture}[scale=0.9]

            \node (a1) at (-110:2) [vertex] {};
            \node (a2) at ({2*cos(160)}, -1.2) [vertex] {};
            \node (a3) at (-70:2) [vertex] {};
            \node (a4) at ({2*cos(160)}, 1.2) [vertex] {};
            \node (a5) at ({2*cos(20)}, -1.2) [vertex] {};
            \node (a6) at (110:2) [vertex] {};
            \node (a7) at ({2*cos(20)}, 1.2) [vertex] {};
            \node (a8) at (70:2) [vertex] {};

            \node (b1) at (-110:2.4) [labell] {$v_5^p$};
            \node (b2) at ({2.4*cos(160)}, -1.2) [labell] {$v_2^p$};
            \node (b3) at (-70:2.4) [labell] {$v_6^p$};
            \node (b4) at ({2.4*cos(160)}, 1.2) [labell] {$v_1^p$};
            \node (b5) at ({2.4*cos(20)}, -1.2) [labell] {$v_8^p$};
            \node (b6) at (110:2.4) [labell] {$v_3^p$};
            \node (b7) at ({2.4*cos(20)}, 1.2) [labell] {$v_7^p$};
            \node (b8) at (70:2.4) [labell] {$v_4^p$};
            
            \draw[opacity=0.55, densely dashed] (a1) to (a3);
            \draw[opacity=0.55, densely dashed] (a2) to (a4);
            \draw[opacity=0.55, densely dashed] (a6) to (a8);
            \draw[opacity=0.55, densely dashed] (a5) to (a7);
            \draw[opacity=0.55, densely dashed] (a5) to (a6);
            \draw[opacity=0.55, densely dashed] (a3) to (a6);
            \draw[opacity=0.55, densely dashed] (a3) to (a4);

            \draw[opacity=0.55, densely dashed] (a2) to [bend right = 20](a4) [red];
            \draw[opacity=0.55, densely dashed] (a7) to [bend right = 20](a5) [red];

            \draw[opacity=0.55, densely dashed] (a1) to (a7) [red];
            \draw[opacity=0.55, densely dashed] (a2) to (a8) [red];
            \draw[opacity=0.55, densely dashed] (a1) to (a8) [red];

            \draw[thick, blue] (a1) to (a5);
            \draw[thick, blue] (a1) to (a3);
            \draw[thick, blue] (a6) to (a8);
            \draw[thick, blue] (a1) to[bend right] (-2,-0.3);
            \draw[thick, blue] (a5) to (a7);
            \draw[thick, blue] (a8) to (a4);

           \node (b8) at (0.8,-1.2) [labell, blue, scale=1.2] {$xy$};
            \node (b8) at (2.2, 0) [labell, blue, scale=1.2] {$e$};
            
        \tikzset{xshift=7.5cm}

            \node (a1) at (-110:2) [vertex] {};
            \node (a2) at ({2*cos(160)}, -1.2) [vertex] {};
            \node (a3) at (-70:2) [vertex] {};
            \node (a4) at ({2*cos(160)}, 1.2) [vertex] {};
            \node (a5) at ({2*cos(20)}, -1.2) [vertex] {};
            \node (a6) at (110:2) [vertex] {};
            \node (a7) at ({2*cos(20)}, 1.2) [vertex] {};
            \node (a8) at (70:2) [vertex] {};

            \node (b1) at (-110:2.4) [labell] {$v_5^p$};
            \node (b2) at ({2.4*cos(160)}, -1.2) [labell] {$v_2^p$};
            \node (b3) at (-70:2.4) [labell] {$v_6^p$};
            \node (b4) at ({2.4*cos(160)}, 1.2) [labell] {$v_1^p$};
            \node (b5) at ({2.4*cos(20)}, -1.2) [labell] {$v_8^p$};
            \node (b6) at (110:2.4) [labell] {$v_3^p$};
            \node (b7) at ({2.4*cos(20)}, 1.2) [labell] {$v_7^p$};
            \node (b8) at (70:2.4) [labell] {$v_4^p$};
            
            \draw[opacity=0.55, densely dashed] (a1) to (a3);
            \draw[opacity=0.55, densely dashed] (a2) to (a4);
            \draw[opacity=0.55, densely dashed] (a6) to (a8);
            \draw[opacity=0.55, densely dashed] (a5) to (a7);
            \draw[opacity=0.55, densely dashed] (a5) to (a6);
            \draw[opacity=0.55, densely dashed] (a6) to (a3);
            \draw[opacity=0.55, densely dashed] (a3) to (a4);

            \draw[opacity=0.55, densely dashed] (a2) to [bend right = 20](a4) [red];
            \draw[opacity=0.55, densely dashed] (a7) to [bend right = 20](a5) [red];

            \draw[opacity=0.55, densely dashed] (a1) to (a7) [red];
            \draw[opacity=0.55, densely dashed] (a2) to (a8) [red];
            \draw[opacity=0.55, densely dashed] (a1) to (a8) [red];

            \draw[thick, blue] (a7) to (-2,0.3);
            \draw[thick, blue] (a1) to (a3);
            \draw[thick, blue] (a6) to (a8);
            \draw[thick, blue] (a1) to[bend right] (-2,-0.3);
            \draw[thick, blue] (a5) to (a7);
            \draw[thick, blue] (a8) to (a4);

            \node (b8) at (-0.15,0.4) [labell, blue, scale=1.2] {$xy$};
            \node (b8) at (2.2, 0) [labell, blue, scale=1.2] {$e$};

        \end{tikzpicture}
        \end{center}
        \caption{The tree $T$ (in blue) on the set $C_p$ obtained before rotating $e$ into $v_5^pv_7^p$ during $\S$ in Claim~\ref{cl:commonCp}. By connectivity, $T$ contains an edge $xy$ connecting $\bigcup_{j>p}C_j$ with the rest of the graph. On the left, $x = v_8^p$ and on the right, $x = v_7^p$. }
        \label{fig:rotintocp}
    \end{figure}
    
    Hence, there is no common chord that is rotated into $v_5^pv_7^p$, and into $v_2^pv_4^p$ by symmetry. Moreover, up to exchanging $T_k$ and $T_k'$, this also proves that there is no common chord that has been added by rotating $v_3^pv_8^p$ or $v_1^pv_6^p$.
\end{proof}

\begin{claim}\label{cl:34Cp}
The rotations of $\S\setminus\S_0$ involving the common chords in $C_p$ have type (3) or (4).
\end{claim}
\begin{proof}
    By symmetry, we consider the case $p<k$ and the common chord $e=v_7^pv_8^p$. Assume by contradiction that $r_p$ has type (1) or (2).
    
    Then, since the only rotations involving $e$ are either $r_p$ or in $\S_0$, $e$ is rotated into a chord of $T_k' \setminus T_k$ and added by rotating a chord of $T_k \setminus T_k'$.
    By Claim~\ref{cl:commonCp}, $e$ is rotated into $v_2^{p+1}v_4^{p+1}$ and added by rotating $v_1^{p+1}v_6^{p+1}$. Moreover, at least one of these two rotations must lie in $\S_0$, say $e\rightsquigarrow v_2^{p+1}v_4^{p+1}$ (the other case being similar). In particular, all the trees obtained during $\S$ contain $v_2^{p+1}v_4^{p+1}$ afterwards. However, this prevents to create again the chord $v_1^{p+1}v_6^{p+1}$, and thus to recreate $e$, a contradiction.
\end{proof}

\begin{claim}
The chords of $T_k \setminus T_k'$ with both endpoints in $C_p$ are rotated into edges with both endpoints in $C_p$.
\end{claim}

\begin{proof}
Assume by contradiction that $\S$ rotates a chord $uv$ of $T_k \setminus T_k'$ into a chord $vw$, where $u,v \in C_p$ and $w \notin C_p$.
Let $T$ be the tree obtained after performing $uv \rightsquigarrow vw$ during $\S$.
By symmetry, say $w$ is in $\bigcup_{j>p}C_j$. We now distinguish two cases depending on the type of $r_p$ ((3) or (4) by Claim~\ref{cl:34Cp}). \smallskip

    \noindent\textbf{Case 1:} $r_p$ has type (3). \\ 
    By Claim~\ref{cl:commonCp}, $v_7^pv_8^p$ is added back by a rotation $r$ in $\S_0$ which removes $v_1^{p+1}v_6^{p+1}$ for the first time in $\S$. 
    In particular, $vw$ cannot be $v_1^{p+1}v_6^{p+1}$, nor can cross $v_1^{p+1}v_6^{p+1}$, hence $w \in \{ v_5^{p+1}, v_6^{p+1} \}$.
    Since $v$ is in $C_p$ and an endpoint of a chord in $T_k \setminus T_k'$, $vw$ either is $v_8^pv_5^p$ or crosses $v_7^pv_8^p$. 
    Thus, $vw$ is not a chord of $T_k'$ and must be removed before performing $r$ during $\S$. 
    However, $vw$ cannot be rotated into an chord of $T_k'$ before performing $r$ during $\S$ (since $v$ is not an endpoint of a chord in $T'_k\setminus T_k$, and the only chords of $T'_k\setminus T_k$ containing $w$ cross $v_1^{p+1}v_6^{p+1}$). 
    Thus, $vw$ is rotated into a chord not in $T_k'$, and we found an additional rotation.\smallskip

    \noindent\textbf{Case 2:} $r_p$ has type (4). \\ 
    By Claim~\ref{cl:commonCp}, $v_7^pv_8^p$ is removed by a rotation $r$ in $\S_0$ which adds $v_2^{p+1}v_4^{p+1}$ for the last time in $\S$. 
    So $vw$ cannot be rotated into $v_2^{p+1}v_4^{p+1}$ and does not cross $v_2^{p+1}v_4^{p+1}$, hence $w \in \{ v_3^{p+1}, v_4^{p+1} \}$. In particular, neither $v$ nor $w$ is an endpoint of a chord of $T_k'\setminus T_k$ except maybe $v_2^{p+1}v_4^{p+1}$. Therefore, $vw$ cannot be rotated into a chord of $T_k' \setminus T_k$ nor into $v_1^pv_2^p$ by choice of $p$.
    
    This shows that $vw$ must be rotated into $v_7^pv_8^p$, hence $v = v_8^p$ and $u = v_3^p$ since $w \notin C_p$ and $uv \in T_k \setminus T_k'$. Since $vw\notin T_k\setminus T'_k$, the rotation $vw\rightsquigarrow v_7^pv_8^p$ is precisely $r_p$. 

    Observe that the rotations $r$, $uv\rightsquigarrow vw$ and $r_p$ must then occur in that order in $\S$. Let $T$ be the tree obtained after performing $r$. By construction, $uv\rightsquigarrow vw$ is a rotation from $\S_0$, hence all the trees obtained until $T$ contain $uv$. In particular, $T$ contains a path connecting $\{u,v\}$ and $\{v_2^{p+1},v_4^{p+1}\}$. Since $T$ is obtained using $r$, this path must be included in $C_{p+1}$ by connectivity, and cannot cross nor contain $v_7^pv_8^p$. The first edge of this path must thus be $v_1^{p+1}z$ for some $z\in\{v_3^{p+1},v_4^{p+1}\}$.

    This edge is not in $T'_k$, hence it has to be removed by a rotation $r'\in\S$. Since every tree obtained after $T$ contain $v_2^{p+1}v_4^{p+1}$, the edge $v_1^{p+1}z$ cannot be rotated into a chord of $T'_k\setminus T_k$ and $r'\notin \S_0$. Moreover, we have $r'\neq r_p$, hence $r'$ is a rotation from $\S\setminus \S_0$ that was not counted by Lemma~\ref{cl:rotk-1}, which concludes.
\end{proof}

Consider the three edges $e_1, e_2, e_3$ obtained by performing the rotations of $\S_0$ which remove the chords of $T_k \setminus T_k'$ with both endpoints in $C_p$ for the first time.
Since these edges have been obtained by rotating chords of $T_k \setminus T_k'$, they are not chords of $T_k' \setminus T_k$.
By Claim~\ref{cl:commonCp}, $e_1,e_2,e_3$ are not common chords either.
And finally, these edges are not common border edges (otherwise, we applied a rotation before that rotated the common border edge into a chord that is not in $T'_k \setminus T_k$).
Thus, these edges have to be removed during $\S$. 
By construction of $C_p$, $e_1,e_2,e_3$ are not rotated into common chords.
Hence, the rotations that remove these edges are in $\S_0$, thus add a chord of $T_k' \setminus T_k$ for the last time.
Since the common chords with endpoints in $C_p$ are either rotated into $v_2^{p+1}v_4^{p+1}$ or $v_5^{p-1}v_7^{p-1}$, or added by rotating $v_1^{p+1}v_6^{p+1}$ or $v_8^{p-1}v_3^{p-1}$, the edges $e_1,e_2,e_3$ cannot be rotated into chords of $T_k' \setminus T_k$ with both endpoints in $C_{p-1}$ or with both endpoints in $C_{p+1}$.
So the edges $e_1, e_2, e_3$ are rotated into distinct chords of $T_k' \setminus T_k$ with both endpoints in $C_p$.
This gives a rotation sequence from $T_k \cap C_p$ to $T_k' \cap C_p$ using $6$ rotations, which contradicts Lemma~\ref{cl:T1rotation}.
\end{proof}

The previous lemmas ensures that a minimal rotation sequence $\S$ contains at least $7k$ rotations, which completes the proof of Theorem~\ref{thm:lb_rotations}.

\paragraph{Acknowledgments.} The first and third authors would like to thank Valentin Gledel for interesting discussions on the problems on an earlier stage of this project.

 \bibliographystyle{plain}
\bibliography{bibfile.bib}
\end{document}